%% file: Linear_DFVS.tex
\documentclass[a4paper,11pt]{article}
\usepackage{listings}
  \usepackage{mathrsfs}
       \usepackage[mathscr]{euscript}

\newcommand{\longversion}[1]{{#1}
\newcommand{\shortversion}[1]}
\usepackage{}
\usepackage{amsfonts,amsthm,amsmath}
\usepackage{amstext}
\usepackage{graphicx,tikz}
\RequirePackage{fancyhdr}
\usepackage{xcolor}
\usepackage{boxedminipage}
\newcommand{\olddaniel}[1]{}

\newcommand{\cQ}{\mathcal{Q}}

\newcommand{\bA}{\mathbf{A}}
\newcommand{\bB}{\mathbf{B}}

\usepackage{listings}

\usepackage{vmargin}
\setmarginsrb{1in}{1in}{1in}{1in}{0pt}{0pt}{0pt}{6mm}

\newcommand{\bran}[1]{branchable\xspace}

\usepackage[linesnumbered, boxed, algosection,noline]{algorithm2e}

\newcommand{\myparagraph}[1]{\smallskip\noindent{\textbf{\sffamily #1} \ }}

\newtheorem{theorem}{Theorem}
\newtheorem{lemma}{Lemma}[section]
\newtheorem{claim}{Claim}[section]
\newtheorem{corollary}{Corollary}
\newtheorem{definition}{Definition}[section]
\newtheorem{observation}{Observation}[section]
\newtheorem{proposition}{Proposition}[section]

\date{}

\usepackage{graphics}

\usepackage{pgf}
\usepackage{paralist} 
\usepackage{graphicx}
\usepackage{boxedminipage}
\usepackage{boxedminipage}
\usepackage{xcolor}
\usepackage[linesnumbered,boxed,algosection]{algorithm2e}
\usepackage{framed}

\newcommand{\NP}{\text{\normalfont  NP}}

\newcommand{\FPT}{\text{\sf FPT}}

\usepackage{amsmath, amssymb, latexsym}
\usepackage{enumerate}

\usepackage{float}
\usepackage{xspace}

\newcommand{\dfvs}{{\sc DFVS}}
\newcommand{\dfas}{{\sc DFAS}}

\newcommand{\sfvs}{{\sc Subset FVS}}

\newcommand{\dfvsfull}{{\sc Directed Feedback Vertex Set}}

\newcommand{\bigoh}{{\mathcal{O}}}

\newcommand{\No}{{\sc No}}
\newcommand{\Yes}{{\sc Yes}}
\newcommand{\Oh}{{\mathcal{O}}}

\usepackage{todonotes}

\usepackage{booktabs}

 \usepackage[ pdftex, plainpages = false, pdfpagelabels, 
                 bookmarks=false,
                 bookmarksopen = true,
                 bookmarksnumbered = true,
                 breaklinks = true,
                 linktocpage,
                 pagebackref,
                 colorlinks = true,  
                 linkcolor = blue,
                 urlcolor  = blue,
                 citecolor = red,
                 anchorcolor = green,
                 hyperindex = true,
                 hyperfigures
                 ]{hyperref}

\usepackage{thm-restate}

\usepackage{thmtools}

\begin{document}


\title{A Linear Time Parameterized Algorithm for {\sc Directed Feedback Vertex Set}}

 \author{
 Daniel Lokshtanov\thanks{University of Bergen, Bergen, Norway. \texttt{daniello@ii.uib.no}}
 \and M. S. Ramanujan\thanks{Algorithms and Complexity Group, TU Wien, Vienna, Austria.     \texttt{ramanujan@ac.tuwien.ac.at}} 
 \and  Saket Saurabh\addtocounter{footnote}{-1}\footnotemark\thanks{The Institute of Mathematical Sciences, Chennai, India. \texttt{saket@imsc.res.in}}
 }
 
%

 \maketitle

\thispagestyle{empty}
\begin{abstract} 
\input{dfvsabstract.tex}
\end{abstract}

\newpage
\pagestyle{plain}
\setcounter{page}{1}

\section{Introduction}
\input{intro.tex}

\section{Preliminaries}\label{sec:prelims}
\input{dfvsprelim.tex}

\section{ The {\FPT} algorithm for {\sc $\cQ$-deletion($\epsilon$) }}\label{sec:main}

\input{full-algo}

\section{Applications}\label{sec:app}
\input{applications}

\section{Proving Lemma \ref{lem:crux}}\label{sec:crux}
\input{section-mainlemma}


\section{Conclusions}
\input{conclusion}

\longversion{
\bibliographystyle{siam}
\bibliography{references}
}

\end{document}

%% file: dfvsabstract.tex

In the {\sc Directed Feedback Vertex Set (DFVS)} problem, the input is a directed graph $D$ on $n$ vertices and $m$ edges, and an integer $k$. The objective is to determine whether there exists a set of at most $k$ vertices intersecting every directed cycle of $D$.  
%
%
Whether or not DFVS admits a fixed parameter tractable (FPT) algorithm was considered the most important open problem in parameterized complexity until Chen, Liu, Lu, O'Sullivan and Razgon [JACM 2008] answered the question in the affirmative. They gave an algorithm for the problem with running time $\Oh(k!4^kk^4nm)$. Since then, no faster algorithm for the problem has been found.
In this paper, we give an algorithm for {\dfvs} with running time $\Oh(k!4^kk^5(n+m))$. Our algorithm is the first algorithm for {\dfvs} with linear dependence on input size. 
Furthermore, the asymptotic dependence of the running time of our algorithm on the parameter $k$ matches up to a factor $k$ the algorithm of Chen, Liu, Lu, O'Sullivan and Razgon.

On the way to designing our algorithm for {\dfvs}, we give a general methodology to shave off a factor of $n$ from iterative-compression based algorithms for a few other well-studied covering problems in parameterized complexity. We demonstrate the applicability of this technique by speeding up by a factor of $n$, the current best {\FPT} algorithms for {\sc Multicut} [STOC 2011, SICOMP 2014] and {\sc  Directed Subset Feedback Vertex Set} [ICALP 2012, TALG 2014]. 

%% file: intro.tex
{\sc Feedback Set} problems are fundamental combinatorial optimization problems. Typically, in these  problems, we are given a graph $G$ (directed or undirected) and a positive integer $k$, and the objective is to select at most $k$ vertices, edges or arcs to hit all cycles of the input graph. 
{\sc Feedback Set} problems are among  Karp's  $21$ \NP-complete problems and  have been topic of active research from 
algorithmic~\cite{BafnaBF99,Bar-YehudaGNR98,BeckerG96,CaoCL10,ChekuriM15,ChenFLLV08,ChenLLOR08,Chitnis:2012DSFVS,CyganNPPRW11,CyganPPW13,EvenNSS98,GuruswamiL15,KawarabayashiK12,KakimuraKK12,KociumakaP14,RamanSS06,Wahlstrom14} as well as structural points of view~\cite{erdHos1965independent,KakimuraKM11,KawarabayashiKKK13,PontecorviW12,ReedRST96,Seymour95,Seymour96}. 
 In particular, such problems constitute one 
  of the most important topics of research in parameterized algorithms~\cite{CaoCL10,ChenFLLV08,ChenLLOR08,Chitnis:2012DSFVS,CyganNPPRW11,CyganPPW13,KawarabayashiK12,KakimuraKK12,KociumakaP14,RamanSS06,Wahlstrom14}, spearheading the development of several new techniques. In this paper, we study
the problem where the objective is to find a set of $k$ vertices that intersects all directed cycles in a given digraph. 
 The problem can be formally stated as follows.


\begin{center}
\begin{boxedminipage}{.95\textwidth}
{\sc Directed Feedback Vertex  Set (\dfvs )}

\begin{tabular}{ r l }
\textit{~~~~Instance:} & A digraph $D$ on $n$ vertices and $m$ edges and a positive integer $k$. \\
\textit{Parameter:} & $k$\\
\textit{Question:} & Does there exist a vertex subset of size at most $k$ that intersects \\ & every cycle in $D$?
\end{tabular}
\end{boxedminipage}
\end{center}
\medskip

For over a decade  \dfvs\ was considered the most important open problem in parameterized complexity.  
In fact, this problem was posed as an open problem in the first few papers on fixed-parameter tractability (FPT)~\cite{DowneyF92,DowneyF95}. In a break-through paper, this problem was shown to be fixed-parameter tractable by 
Chen, Liu, Lu, O'Sullivan and Razgon \cite{ChenLLOR08} 
in 2008. They gave an algorithm that runs in time $\Oh(4^kk!k^4n^4)$ where $n$ is the number of vertices in the input digraph. Subsequently, it was observed that, in fact, the running time of this algorithm is $\Oh(4^kk!k^4nm)$ (see for example \cite{CyganFKLMPPS15}). 
Since this break-through, the techniques used to solve \dfvs\ have found numerous applications, yet, \dfvs\ itself has seen no progress since then.


In this paper we make  first progress on  \dfvs\ and obtain the first linear-time {\FPT}-algorithm for \dfvs. 
In particular we give 
the following theorem. 
\begin{restatable}{theorem}{dfvstheorem}\label{thm: dfvs is dfpt}
There is an algorithm for \dfvs\ running in time $\Oh(k!4^kk^5\cdot (n+m))$. 
\end{restatable}
\noindent 
Our algorithm achieves the best  possible dependence on the input size while matching the current best-known parameter-dependence -- that of the algorithm of Chen et al.~\cite{ChenLLOR08}, up to a $O(k)$ factor.
Since it is well known that {\dfvs} cannot be solved in time $2^{o(k)}n^c$ for any constant $c$ under the Exponential Time Hypothesis (ETH) \cite{CyganFKLMPPS15,DehneFLRS07}, our algorithm is in fact \emph{nearly-optimal}.  
%
An alternate perspective on our results, is as a generalization of the well known linear-time algorithms~\cite{Knuth73a,Tarjan72} to {\em recognize directed acyclic graphs} (DAGs), to recognizing digraphs that are at most ``$k$ vertices away'' from being acyclic. 
The algorithms to recognize DAGs are simple and elegant. Therefore it is striking that no linear time algorithm to recognize digraphs that are just one vertex away from being acyclic were known prior to this work. Finally, our algorithm only relies on basic algorithmic and combinatorial tools. Thus, our algorithm does not have huge hidden constants in the running time, and we expect it to be implementable and to perform well in practice for small values of $k$.

\begin{table}[t]
\centering
\setlength{\tabcolsep}{4pt}
{\footnotesize
\begin{tabular}{l l  l}
\toprule
Problem Name       &  Running Time    &  Comment       \\
\midrule
{\sc Treewidth} &  $2^{\Oh(k \log k)}\cdot (n\log n)$    &   STOC' 92,    $8$-approximation  \cite{Reed92} \\
{\sc Treewidth}  &  $2^{\Oh(k^3 \log k)}\cdot n$ &   STOC' 93,    \cite{stocBodlaender93,Bodlaender96}               \\
{\sc Treewidth} &  $\Oh(c^k\cdot n)$    &  FOCS' 13,  $5$-approximation \cite{BodlaenderDDFLP13}     \\
{\sc Crossing Number} &  $f(k)\cdot n^2$    & STOC' 01,  \cite{stocGrohe01,Grohe04}       \\
{\sc Crossing Number} &  $f(k)\cdot n$    &  STOC' 07,  \cite{KawarabayashiR07}       \\
{\sc Vertex Planarization} &  $f(k)\cdot n$    &  FOCS' 09,  \cite{Kawarabayashi09}    \\
{\sc Vertex Planarization} &   $2^{\Oh(k \log k)}\cdot n$ & SODA' 14,  \cite{JansenLS14}      \\
{\sc Odd Cycle Transversal} &  $f(k)\cdot (n+m)\alpha(n+m)$    &   SODA' 10,  \cite{KawarabayashiR10}       \\
{\sc Odd Cycle Transversal} &  $\Oh(4^k\cdot k^{\Oh(1)} \cdot (n+m))$    & SODA' 14,   \cite{IwataOY14,RamanujanS14}      \\
{\sc Genus} &  $\Oh(c^k\cdot n)$    & FOCS' 08,   \cite{KawarabayashiMR08}      \\
{\sc Interval Vertex Deletion} &  $\Oh(8^k\cdot (n+m))$    &  SODA' 16,  \cite{Cao16}     \\
{\sc Planar-$\cal F$-Deletion} &          $f(k)\cdot n^2$     &   JACM' 88,   \cite[Theorem $6$]{FellowsL88} \\
{\sc Planar-$\cal F$-Deletion} &            $f(k)\cdot n$    &   STOC' 93,   \cite[Theorem $6.1$]{stocBodlaender93} \cite[Theorem $7.1$]{Bodlaender96}   \\
{\sc Planar-$\cal F$-Deletion} &            $\Oh(c^k\cdot n)$     &  FOCS' 12 , Randomized, \cite{FominLMS12}  \\
{\sc Planar-$\cal F$-Deletion} &            $\Oh(c^k\cdot n)$     &  SODA' 15, Deterministic, \cite{FominLMRS15}    \\
{\sc Graph Minor Decomposition} &  $f(k)\cdot n^2$    & SODA' 13,  \cite{GroheKR13}      \\
{\sc Permutation Pattern} &            $2^{\Oh(k^2 \log k)}\cdot n$       &  SODA' 14, \cite{GuillemotM14}    \\
\bottomrule
\end{tabular}
		
	}
	\caption{\label{fig:vertexresults}Summary of some new and old  parameterized algorithms with emphasis on improving the dependence on the input size.}
\end{table}

\medskip
\noindent 
{\bf Dependence on input size in {\FPT} algorithms.}  Our algorithm for \dfvs\  belongs to a large body of work where the main goal is to design linear time algorithms  for \NP-hard problems for a fixed value of $k$. That is, to design  an algorithm with running time $f(k)\cdot \Oh(|I|)$, where $|I|$ denotes the size of the input instance.  This area of research predates even parameterized complexity. The genesis of parameterized complexity is in the theory of graph minors, developed by Robertson and Seymour~\cite{RobertsonS95b,RobertsonS03b,RobertsonS04}. Some of the important algorithmic consequences of this theory include $\Oh(n^3)$ algorithms for {\sc Disjoint Paths} and {\sc $\cal F$-Deletion} for every fixed values of $k$. 
These results led to a whole new area of designing algorithms for \NP-hard problems with as small dependence on the input size as possible; resulting in 
algorithms with improved dependence on the input size for {\sc Treewidth}~\cite{stocBodlaender93,Bodlaender96}, \FPT{} approximation for 
{\sc Treewidth}~\cite{BodlaenderDDFLP13,Reed92}
{\sc Planar $\cal F$-Deletion}~\cite{stocBodlaender93,Bodlaender96,FellowsL88,FominLMS12,FominLMRS15}, and {\sc Crossing Number}~\cite{stocGrohe01,Grohe04,KawarabayashiR07}, to name a few.

The advent of parameterized complexity started to shift the focus away from the running time dependence on input size to the dependence on the parameter. That is, the goal became designing parameterized algorithms with running time upper bounded by $f(k)n^{\Oh(1)}$, where the function $f$ grows as slowly as possible.
Over the last two decades researchers have tried to optimize one of these objectives, but rarely both at the same time. More recently, efforts have been made towards obtaining linear (or polynomial) time parameterized algorithms that compromise as little as possible on the dependence of the running time on the parameter $k$. The gold standard for these results are algorithms with linear dependence on input size as well as provably optimal (under ETH) dependence on the parameter. New results in this direction include parameterized algorithms for problems such as {\sc Odd Cycle Transversal}~\cite{IwataOY14,RamanujanS14}, {\sc Subgraph Isomorphism}~\cite{Dorn10}, {\sc Planarization}~\cite{JansenLS14,Kawarabayashi09}, {\sc Subset Feedback Vertex Set}~\cite{LokshtanovRS15}
as well as a single-exponential and linear time parameterized constant factor approximation algorithm for {\sc Treewidth}~\cite{BodlaenderDDFLP13}. Other recent results include parameterized algorithms with improved dependence on input size for a host of problems~\cite{GroheKR13,KawarabayashiKR12,KawarabayashiM08,KawarabayashiMR08,KawarabayashiR09,KawarabayashiR10}.  We refer to Table~\ref{fig:vertexresults} for a quick overview of results in this area.

\medskip

\noindent 
{\bf Methodology.}
\input{methodology}

%% file: methodology.tex
At the heart of numerous {\FPT}-algorithms lies the fact that, if one could efficiently compute a sufficiently good approximate solution,
it is then sufficient to design an {\FPT}-algorithm for the ``compression version'' of a problem in order to obtain an {\FPT}-algorithm for the general version. 
In the compression version of a problem, the input also includes an approximate solution whose size depends only on the parameter. Since a given approximate solution may be used to infer significant structural information about the input, it is usually much easier to design {\FPT}-algorithms for the compression version than for the original problem. The efficiency of this approach clearly depends on two factors --  (a) the time required to compute an approximate solution and (b) the time required to solve the compression version of the problem when the approximate solution is provided as input.
  
  This approach has been used mainly in the following two settings. In the first setting, the objective is the design of linear-time {\FPT}-algorithms. In this setting, for certain problems, it can be shown that if the treewidth of the input graph is bounded by a function of the parameter then the problem can be solved by a linear-time {\FPT}-algorithm (either designed explicitly or obtained by invocation of an appropriate algorithmic meta-theorem). On the other hand, if the treewidth of the input graph exceeds a certain bound, then there is a sufficiently large (induced) matching which one can contract and obtain an instance whose \emph{size} is a constant fraction of that of the original input. Now, the algorithm is recursively invoked on the reduced instance and certain problem-specific steps are used to convert the recursively computed solution into an approximate solution for the given instance. Then, a linear-time {\FPT}-algorithm for the compression version is executed to solve the general problem on this instance. Some of the results that fall under this paradigm are Bodlaender's linear {\FPT}-algorithm for {\sc Treewidth} \cite{Bodlaender96}, the  {\FPT}-approximation algorithms for {\sc Treewidth}~\cite{BodlaenderDDFLP13,Reed92}, as well as   algorithms for {\sc Vertex Planarization}~\cite{Kawarabayashi09,JansenLS14}. Let us call this the {\em method of global shrinking}. This is one of the most commonly used techniques in designing linear-time FPT-algorithms on undirected graphs.  
  
  On the other hand, when designing {\FPT}-algorithms where the dependence on the input is \emph{not} required to be linear, one can use the \emph{iterative}-compression technique, introduced by Reed, Smith and Vetta \cite{ReedSV04}. 
Here the input instance is gradually built up by simple operations, such as vertex additions. After each operation, an optimal solution is re-computed, starting from an optimal solution to the smaller instance.
%
%
By its very definition, the iterative compression technique does not lend itself to the design of linear-time {\FPT}-algorithms. Hence, it may look as if one has to look for alternative ways when aiming for linear-time {\FPT}-algorithms. In the recent years, some of the problems which were initially solved using the iterative compression technique, have seen the development of entirely new algorithms.
%
%
Examples include the first linear-time {\FPT}-algorithms for the {\sc Odd Cycle Transveral}, {\sc Almost 2-SAT} and {\sc Edge Unique Label Cover} problems~\cite{IwataOY14,RamanujanS14,YoichiWY13}. All of these algorithms are based on branching and linear programming techniques.

  Another general approach to the design of linear-time {\FPT}-algorithms has been introduced by Marx et al. \cite{MarxOR13}. These algorithms are based on the ``Treewidth Reduction Theorem''' which states that in undirected graphs, for any pair of vertices $s$ and $t$, all minimal $s$-$t$ separators of bounded size are contained in a part of the graph that has bounded treewidth.

However, all of these techniques are specifically designed for undirected graphs and hence fail when addressing problems on directed graphs. Our main contribution is a novel approach for `lifting' linear-time {\FPT}-algorithms for the compression version of feedback-set problems on \emph{digraphs} to linear-time {\FPT}-algorithms for the general version of the problem. One may think of our approach as a generalization of the method of global shrinking, pioneered by Bodlaender~\cite{Bodlaender96} in his celebrated linear time {\FPT}-algorithm for {\sc Treewidth}. 
 
%
%
%
%
%
%

Given a digraph $D$, we say that $S$ is a {\em directed feedback vertex set (dfvs)} if deleting $S$ from $D$ results in a DAG. At the core of our algorithm lies the following new structural lemma regarding digraphs with a small dfvs. 

 \begin{restatable}{lemma}{cruxlemmadfvs}
\label{lem:cruxdfvs}
Let $D$ be a strongly connected digraph and $p\in {\mathbb N}$. There is an algorithm that, given $D$ and $p$, runs in time $\bigoh(p^2m)$ (where $m$ is the number of arcs in $D$) and either correctly concludes that $D$ has no dfvs of size at most $p$ or returns a set $S$ with at most $2p+1$ vertices such that one of the following holds.
\begin{itemize}
\setlength{\itemsep}{-2pt}
\item $S$ is a dfvs for $D$.
\item $D-S$ has at least 2 non-trivial strongly connected components (strongly-connected components with at least 2 vertices).
\item The number of arcs of $D$ whose head and tail occur in the same  
 non-trivial strongly connected component of $D-S$ (arcs participating in a cycle of $D-S$) is at most $\frac{m}{2}$.
\item  If $D$ has a dfvs of size at most $p$ then $D-S$ has a dfvs of size at most $p-1$.
\end{itemize}
%
\end{restatable}

 \noindent 
Our linear-time FPT algorithm for {\dfvs}  is obtained by a careful interleaving of the algorithm of Lemma \ref{lem:cruxdfvs} with an algorithm solving the compression version of {\dfvs} (in this case, the compression routine of Chen et al. \cite{ChenLLOR08}). The proof of Lemma \ref{lem:cruxdfvs} itself is based on extending the notion of important sequences~\cite{LokshtanovR12} to digraphs, and then analyzing a single such sequence. Furthermore, the proof of Lemma \ref{lem:cruxdfvs} only relies on properties of {\dfvs} that are shared by several other feedback set and graph separation problems. Hence, we directly prove a more general version of this lemma and show how it can be used as a black box to shave off a factor of $n$ from existing iterative compression based algorithms for other problems which satisfy certain conditions. This results in speeding up by a factor of $n$, the current best {\FPT} algorithms for 
{\sc Multicut}~\cite{MarxRazmulticut,MarxR14,BousquetDT11} and {\sc  Directed Subset Feedback Vertex Set}~\cite{ChitnisCHM15,Chitnis:2012DSFVS}.

%% file: dfvsprelim.tex
\myparagraph{Parameterized Complexity.} 
Formally, a {\em parameterization} of a problem is the assignment of an integer $k$ to each input instance and  we say that a parameterized problem is {\em fixed-parameter tractable} ({\FPT}) if there is an algorithm that solves the problem in time $f(k)\cdot |I|^{\bigoh(1)}$, where $|I|$ is the size of the input instance and $f$ is an
arbitrary computable function depending only on the parameter $k$. 
For more background, the reader is referred to the monographs \cite{DF99,FG06,Nie06,CyganFKLMPPS15}.

\medskip 
\myparagraph{Digraphs.}
%
%
For a digraph $D$ and vertex set $X\subseteq V(D)$, we say that $X$ is a \textbf{dfvs} of $D$ if $X$ intersects every cycle in $D$. We say that $X$ is a minimal dfvs of $D$ if no proper subset of $D$ is also a dfvs of $D$. We call $X$ a minimum dfvs of $D$ if there is no smaller dfvs of $D$. For an arc $(u,v)\in A(D)$, we refer to $u$ as the \textbf{tail} of the arc and $v$ as the \textbf{head}.
 For a subset $X$ of vertices, we use $N^+(X)$ to denote the set of out-neighbors of $X$ and $N^-(X)$ to denote the set of in-neighbors of $X$. We use $N^i[X]$ to denote the set $X\cup N^i(X)$ where $i\in \{+,-\}$. We denote by $A[X]$ the subset of $A(D)$ with both endpoints in $X$. A strongly connected component of $D$ is a maximal subgraph in which every vertex has a directed path to every other vertex. We say that a strongly connected component is \textbf{non-trivial} if it has at least 2 vertices and {\bf trivial} otherwise. For disjoint vertex sets $X$ and $Y$, $Y$ is said to be reachable from $X$ if for \emph{every} vertex  $y\in Y$, there is a vertex $x\in X$ such that the digraph contains a directed path from $x$ to $y$. 

 \myparagraph{Structures.}
 	For $\epsilon\in {\mathbb N}$, an $\epsilon$-structure is a tuple where the first element of the tuple is a digraph $D$ with the remaining elements of the tuple being relations of arity at most $\epsilon$ over $V(D)$. Two $\epsilon$-structures $Q_1$ and $Q_2$ are said to have the same type if $Q_1=(D_1,R_1,\dots, R_\ell)$, $Q_2=(D',R_1',\dots, R_\ell')$ and for each $i\in [\ell]$, $R_i$ and $R_i'$ are relations of the same arity.
The size of an $\epsilon$-structure $Q=(D,R_1,\dots, R_\ell)$ is denoted as $|Q|$ and is defined as $m+n+ \epsilon \cdot \Sigma_{i=1}^\ell |R_i|$, where $m$ and $n$ are the number of vertices in $D$ and $|R_i|$ is the number of tuples in $R_i$. In this paper, whenever we talk about a family $\cQ$ of $\epsilon$-structures, it is to be understood that $\cQ$ only contains $\epsilon$-structures which are pairwise of the same type and this type is also called the {\em type of $\cQ$}.

\begin{definition}
	Let $Q=(D,R_1,\dots, R_\ell)$ be an $\epsilon$-structure. For a set $X\subseteq V(D)$, let $D_X=D[X]$. 
	We define the induced structure $Q[X]=(D[X],R_1|_X,\dots, R_\ell|_X)$ where $R_i|_X$ is the restriction of the relation $R_i$ to the set $V(D_X)$, that is those tuples in $R_i$ which have all elements in $X$. For any $X\subseteq V(D)$, we denote by $Q-X$ the substructure $Q[V(D)\setminus X]$.
\end{definition}

\begin{definition}
	Let $\cQ$ be a family of $\epsilon$-structures. We say that $\cQ$ is {\bf hereditary} if for every $Q\in \cQ$, every induced substructure of $Q$ is also in $\cQ$. We say that a family $\cQ$ of $\epsilon$-structures is {\bf linear-time recognizable} if there is an algorithm that, given an $\epsilon$-structure $Q$, runs in time $\bigoh(|Q|)$ and correctly decides whether $Q\in \cQ$. Finally, we say that $\cQ$ is {\bf rigid} if the following two properties hold:\begin{itemize}	\item For every $\epsilon$-structure $Q=(D,R_1,\dots, R_\ell)$, if $D$ has no arcs then   $Q\in \cQ$ and 
	\item  $Q=(D,R_1,\dots, R_\ell)\in \cQ$ if and only if for every strongly connected component $C$ in the digraph $D$, the induced substructure $Q[C]\in \cQ$.
	\end{itemize}
	\end{definition}


The {\sc $\cQ$-Deletion($\epsilon$)} problem is formally defined as follows.

\begin{center}
\begin{boxedminipage}{.96\textwidth}
{\sc $\cQ$-Deletion($\epsilon$)}

\begin{tabular}{ r l }
\textit{~~~~Instance:} & An $\epsilon$-structure $Q=(D,R_1,\dots, R_\ell)$ and a positive integer $k$. \\
\textit{Parameter:} & $k$\\
\textit{Question:} & Does there exist a set $X\subseteq V(D)$ of size at most $k$ such that $Q-X\in \cQ$?
\end{tabular}
\end{boxedminipage}
\end{center}
\medskip

Our main contribution is a theorem (Theorem \ref{thm:meta}) that, under certain conditions which are fulfilled by several well-studied special cases of {\sc $\cQ$-Deletion($\epsilon$)},  guarantees an {\FPT} algorithm for {\sc $\cQ$-Deletion($\epsilon$)} whose running time has a specific form.

A set $X\subseteq V(D)$ such that $Q-X\in \cQ$ is called a \emph{deletion set of $Q$ into $\cQ$}. In the {\sc $\cQ$-Deletion($\epsilon$) Compression} problem, the input is a triple $(Q,k,\hat W)$ where $(Q,k)$ is an instance of {\sc $\cQ$-Deletion($\epsilon$)} and $\hat W$ is a vertex set such that  $Q-\hat W\in \cQ$. The question remains the same as for {\sc $\cQ$-Deletion($\epsilon$)}. 
However, the parameter for this problem is $k+|\hat W|$.
We say that an algorithm $\bA$ is an algorithm for the the {\sc $\cQ$-Deletion($\epsilon$) Compression} problem if, on input $Q,k,\hat W$ the algorithm either correctly concludes that $(Q,k)$ is a {\No} instance of  {\sc $\cQ$-Deletion($\epsilon$)} or \emph{computes} a smallest set $X$ of size at most $k$ such that $Q-X\in \cQ$.


%% file: full-algo.tex

%
%
%
%
%


In this section, we formally state and prove our main theorem. In the next section, we demonstrate how a direct application of this theorem speeds up by a factor of $n$, existing {\FPT} algorithms for certain well-studied feedback set and graph separation problems.


 \begin{theorem}\label{thm:meta}
 	Let $\epsilon\in \mathbb N$ and let $\cQ$ be a linear-time recognizable,  hereditary and rigid family of 
$\epsilon$-structures. Let $\gamma\in {\mathbb N}, d\in {\mathbb R}_{>1}$ and $f:{\mathbb N}\to {\mathbb N}$ such that $f(t)\geq t$ and  $f(t-1)\leq \frac{f(t)}{d}$ for every $t\in {\mathbb N}$. 
 	\begin{itemize} 
\setlength{\itemsep}{-2pt}
\item  Let ${\bA}$ be an algorithm for {\sc $\cQ$-Deletion($\epsilon$) Compression} 
that, on input $Q=(D,R_1,\dots, R_\ell)$, $k$ and $\hat W$, runs in time $\bigoh(f(k)\cdot |Q|^\gamma \cdot |\hat W|)$, 
where 
	 $\hat W$ is a deletion set of $Q$ into $\cQ$,

\item Let ${\bB}$ be an algorithm that, on input $Q=(D,R_1,\dots, R_\ell)\notin \cQ$, runs in time $\bigoh(|Q|)$ and returns a pair of vertices $u,v$ such that every deletion set of $Q$ into $\cQ$ which is disjoint from $u$ and $v$ is a $u$-$v$ separator in $D$. 
\end{itemize}
Then, there is an algorithm that, given an instance $(Q=(D,R_1,\dots, R_\ell),k)$ of {\sc $\cQ$-Deletion($\epsilon$)} and the algorithms $\bA$ and $\bB$, runs in time $\bigoh(f(k)\cdot k \cdot |Q|^\gamma)$ and either computes a set $X$ of size at most $k$ such that $Q-X\in \cQ$ or correctly concludes that no such set exists.
 	
 \end{theorem}

Before we proceed, we make a few remarks regarding the conditions in the premise of the lemma. Note that we require the running time of Algorithm $\bA$ to be of the form $\bigoh(f(k)\cdot |Q|^\gamma \cdot |\hat W|)$ in spite of the {\sc $\cQ$-deletion($\epsilon$) compression} problem being formally parameterized by $|\hat W|+k$. At first glance, it may appear that this is a requirement that is much stronger than simply asking for an {\FPT} algorithm for {\sc $\cQ$-deletion($\epsilon$) compression}. However, we point out that as long as $\cQ$ is hereditary, this requirement is in fact no stronger than simply asking for an {\FPT} algorithm for {\sc $\cQ$-deletion($\epsilon$) compression}. Precisely, if there is an {\FPT} algorithm for {\sc $\cQ$-deletion($\epsilon$) compression}, that is an algorithm that runs in time $\bigoh(g(k+|\hat W|)\cdot |Q|^\delta)$ for some function $g$ and constant $\delta$, then we can obtain an algorithm for {\sc $\cQ$-deletion($\epsilon$) compression} that runs in time $\bigoh(g(2k+1)\cdot |Q|^\delta \cdot |\hat W|)$ by using a folklore trick of running the compression step for the special case of $|\hat{W}|=k+1$, $|\hat W|$  times.

The main technical component of the proof of this theorem is a generalization of Lemma \ref{lem:cruxdfvs}. The proof of this lemma (Lemma \ref{lem:crux}), is fairly technical and requires the introduction of more notation. In order to keep the presentation of the paper streamlined, we only state the lemma in this section and describe how we use it in the proof of Theorem \ref{thm:meta} with the proof of the lemma postponed to Section \ref{sec:crux}.


 \begin{restatable}{lemma}{cruxlemma}
\label{lem:crux}
Let $\epsilon\in \mathbb N$ and let $\cQ$ be a linear-time recognizable,  hereditary  and rigid family of 
$\epsilon$-structures.
 There is an algorithm that, given an $\epsilon$-structure $Q=(D,R_1,\dots, R_\ell)\notin \cQ$ where $D$ is strongly connected, vertices $u,v\in V(D)$, and $p\in {\mathbb N}$, runs in time $\bigoh(p^2|Q|)$ (where $m$ is the number of arcs in $D$) and either correctly concludes that $D$ has no $u$-$v$ separator of size at most $p$ or returns a set $S$ with at most $2p+2$ vertices such that one of the following holds.
\begin{itemize}
\setlength{\itemsep}{-2pt}
\item $Q-S\in \cQ$.
\item $D-S$ has at least 2 strongly connected components each of which induces a substructure of $Q$ not in $\cQ$.
\item The strongly connected components of $D-S$ can be partitioned into 2 sets inducing substructures of $Q$, say $Q_1$ and $Q_2$ such that $Q_1\notin \cQ$, $Q_2\in \cQ$ and  $|Q_1|\leq \frac{1}{2}|Q|$.
\item  If $Q$ has a deletion set of size at most $p$ into $\cQ$ then $Q-S$ has a deletion set of size at most $p-1$ into $\cQ$.
\end{itemize}
%
\end{restatable}

We now return to Theorem \ref{thm:meta} and proceed to prove it assuming this lemma as a black-box. We describe our algorithm for {\sc $\cQ$-deletion($\epsilon$)} using the algorithms $\bA$, $\bB$ and the algorithm of Lemma \ref{lem:crux} as subroutines. The input to the algorithm in Theorem \ref{thm:meta} is an instance $(Q=(D,R_1,\dots, R_\ell),k)$ of {\sc $\cQ$-deletion($\epsilon$)} and the output is {\No} if $Q$ has no deletion set into $\cQ$ of size at most $k$ and otherwise, the output is a set $X$ which is a \emph{minimum size} deletion set into $\cQ$ of $Q$ of size at most $k$. 

\medskip
\noindent
\textbf{Description of the Algorithm of Theorem \ref{thm:meta} and Correctness.}
We now give a formal description of the algorithm. The algorithm is recursive, each call takes as input an $\epsilon$-structure  $Q=(D,R_1,\dots, R_\ell)$ and integer $k$. In the course of describing the algorithm we will also prove by induction on $k+|Q|$ that the algorithm either correctly concludes that $Q$ has no deletion set into $\cQ$ of size at most $k$, or finds a minimum size deletion set of $Q$ into $\cQ$, say $X$ of size at most $k$. The algorithm proceeds as follows.

In time linear in the size of the digraph $D$, the algorithm computes the decomposition of $D$ into strongly connected components. Let $D'$ be the digraph obtained from $D$ by removing from $D$ all strongly connected components which  induce a substructure of $Q$ that is already in $\cQ$. This operation is safe because the class $\cQ$ is rigid and hereditary. That is, if $Q'=(D',R_1',\dots, R_\ell')$ is the substructure of $Q$ induced on $V(D')$ then any deletion set of $Q$ into $\cQ$ is a deletion set of $Q'$ into $\cQ$ and vice versa. So the algorithm proceeds by working on $Q'$ instead. For ease of description, we now revert back to the input $\epsilon$-structure $Q=(D,R_1,\dots, R_\ell)$ and assume without loss of generality that $D$ does not contain any trivial strongly connected components.

If $D$ is the empty graph or more generally, if $Q\in \cQ$, then the algorithm correctly returns the empty set as a minimum size deletion set of $Q$ into $\cQ$. From now on we assume that $D$ is non-empty. Since $D$ does not contain any trivial strongly connected components this implies that $m \geq n \geq 3$ and hence $|Q|\geq 3$. 

If $k=0$ the algorithm correctly returns {\No}, since $Q\notin \cQ$. From now on we assume that $k \geq 1$. 
For $k \geq 1$, we determine from the computed decomposition of $D$ into strongly connected components whether $D$ is strongly connected. If it is not, then let $C$ be the vertex set of an arbitrarily chosen strongly connected component of $D$. The algorithm calls itself recursively on the instances $(Q[C], k - 1)$ and $(Q - C,k - 1)$. If either of the recursive calls return {\No} the algorithm returns {\No} as well since, both $Q[C]$ and $Q - C$ need to contain at least one vertex from any deletion set of $Q$ into $\cQ$. Otherwise the recursive calls return sets $X_1$ and $X_2$ such that $X_1$ is a deletion set of $Q[C]$ into $\cQ$, $X_2$ is a deletion set of $Q - C$ into $\cQ$ and both $X_1$ and $X_2$ have size at most $k-1$ each. The algorithm executes Algorithm $\bA$ on ($Q$, $k$) with $\hat{W} = X_1 \cup X_2$, and correctly returns the same answer as the Algorithm $\bA$. From now on we assume that $D$ is strongly connected.

For $k \geq 1$ and strongly connected graph $D$ the algorithm proceeds as follows. It starts by running the algorithm $\bB$ on $Q$ to compute in time $\bigoh(|Q|)$ a pair of vertices $u,v\in V(D)$ such that \emph{every} deletion set of $Q$ into $\cQ$ which is disjoint from $u$ and $v$ hits all $u$-$v$ paths in $D$. Clearly, $Q,u,v$ satisfy the premise of Lemma \ref{lem:crux}. Hence we execute the subroutine described in Lemma \ref{lem:crux} on $Q,u,v$ with $p = k$. Recall that the execution of this subroutine will have one of two possible outcomes. In the first case, the subroutine returns a set $S\subseteq V(D)$ of size at most $2k+2 \leq 3k$ satisfying one of the properties in the statement of Lemma~\ref{lem:crux}. In the second case, the subroutine concludes that $D$ has no $u$-$v$ separator of size at most $p$. But in this case, we infer that 
$Q$ has no deletion set into $\cQ$ of size at most $k$ \emph{disjoint from} $\{u,v\}$ and hence we define $S$ to be the set $\{u,v\}$. Now, observe that this set $S$ trivially satisfies the last property in the statement of Lemma~\ref{lem:crux}. Hence,  irrespective of the outcome of the subroutine, we will have computed a set $S$ of size at most $3k$ which satisfies one of the four properties in the statement of Lemma \ref{lem:crux}.  

Observe that it is straightforward to check in linear time whether $S$ satisfies any of the first 3 properties. Therefore, if none of these properties are satisfied, then we assume that $S$ satisfies the last property. Furthermore,  we work with the earliest property that $S$ satisfies. That is, if $S$ satisfies Property $i$ and Property $j$ where $1\leq i<j\leq 4$ then we execute the steps corresponding to Case $i$. Subsequent steps of our algorithm will depend on the output of this check on $S$.

\begin{description}
\item[Case 1:]\emph{$Q-S\in \cQ$.} 
In this case, we 
execute Algorithm $\bA$ 
on $Q$, $k$, with $\hat{W} = S$ to either conclude that $Q$ has no deletion set into $\cQ$ of size at most $k$, in which case we return {\No}, or obtain a minimum size set $X$ which has size at most $k$ and is a deletion set of $Q$ into $\cQ$. In this case we return $X$.

\item[Case 2:]\emph{$D-S$ has at least 2 non-trivial strongly connected components each of which induces a substructure of $Q$ not in $\cQ$.} Let $C$ be one such non-trivial strongly connected component of $D-S$. We know that any deletion set of $Q$ into $\cQ$ must contain at least one vertex in $C$ and at least one vertex in $D - (S \cup C)$. Hence any deletion set of $Q$ into $\cQ$ of size at most $k$ must contain at most $k-1$ vertices in $C$ and at most $k-1$ vertices in $D - (S \cup C)$. Thus, the algorithm solves recursively the instances $(Q[C], k - 1)$ and $(Q - (C \cup S), k - 1)$. If either of the the recursive calls return {\No} the algorithm correctly returns {\No} as well. Otherwise the recursive calls return vertex sets $X_1$ and $X_2$ such that $X_1$ is a deletion set of $Q[C]$ into $\cQ$, $X_2$ is a deletion set of  $Q - (C \cup S)$ into $\cQ$, and both $X_1$ and $X_2$ have size at most $k-1$ each. The algorithm then calls the 
Algorithm $\bA$
on $Q$, $k$ with $\hat{W} = X_1 \cup X_2 \cup S$, and correctly returns the same answer as the Algorithm $\bA$.

\item[Case 3:]\emph{The strongly connected components of $D-S$ can be partitioned into 2 sets inducing substructures of $Q$, say $Q_1$ and $Q_2$ such that $Q_1\notin \cQ$, $Q_2\in \cQ$ and  $|Q_1|\leq \frac{1}{2}|Q|$.} 
Observe that since $S$ did not fall into the earlier cases, we may assume that $S$ is \emph{not} a deletion set of $Q$ into $\cQ$ and  $D - S$ has at most 1 non-trivial strongly connected component. Thus $D-S$ has exactly one non-trivial strongly connected component $C$ which induces a structure not in $\cQ$, and this component induces a structure of size at most $\frac{1}{2}|Q|$. We  recursively invoke the algorithm on input $(Q[C],k)$. If the recursive invocation returned {\No}, then it follows that $Q$ does not have a deletion set into $\cQ$ of size at most $k$, so we can return {\No} as well. On the other hand, if the recursive call returned a set $X$ which is a deletion set of $Q[C]$ into $\cQ$ of size at most $k$ then $S \cup X$ is a deletion set of $Q$ into $\cQ$ of size at most $4k$. Now, we execute Algorithm $\bA$ on $Q$, $k$  with $\hat{W} = S \cup X$ and return the same answer as the output of this algorithm

\item[Case 4:]\emph{If $Q$ has a deletion set into $\cQ$ of size at most $k$ then $Q-S$ has a deletion set into $\cQ$ of size at most $k-1$.}
Recall that we arrive at this case only if the other cases do not occur. We recursively invoke the algorithm on the instance $(Q-S,k-1)$. If the recursion concluded that $Q-S$ does not have a deletion set into $\cQ$ of size at most $k-1$, then we correctly return that $Q$ has no deletion set into $\cQ$ of size at most $k$. Otherwise, suppose that the recursive call returns a set $X$ which is a deletion set of $Q - S$ into $\cQ$ of size at most $k-1$. Now, $S \cup X$ is a deletion set of $Q$ into $\cQ$ of size at most $4k$. Hence, we execute Algorithm $\bA$ on $Q,k$ with $\hat{W} = S \cup X$ and return the same answer the output of this algorithm.
\end{description}
	
Whenever the algorithm makes a recursive call, either the parameter $k$ is reduced to $k-1$ or 
the size of the substructure the algorithm is called on is smaller than $Q$. Thus the correctness of the algorithm and the fact that the algorithm terminates follows from induction on $k+|Q|$. 

\smallskip

\noindent
\textbf{Running Time analysis.} 
We now analyse the running time of the above algorithm when run on an instance $(D,k)$ in terms of the parameters $k$, $n$ and $m$. Before proceeding with the analysis, let us fix some notation. In the remainder of this section, we set
\begin{itemize}\setlength\itemsep{-1.5mm}

\item $\alpha$ to be a constant such that Algorithm $\bA$ on input $Q,k,\hat W$ runs in time $\alpha f(k)\cdot  |Q|^\gamma\cdot |\hat{W}|$,
\item $\beta$ be a constant so that computing the decomposition of $D$ into strongly connected components, removing all trivial strongly connected components, running the algorithm of Lemma~\ref{lem:crux}, then determining which of the four cases apply, and then outputting the substructure induced by a strongly connected component of $D-S$ such that this substructure is not in $\cQ$, takes time $\beta \cdot  k^2\cdot |Q|$.
\end{itemize}

Based on $\alpha$ and $\beta$ we pick a constant $\mu$ such that $\mu \geq \max\left\{20 \beta, \frac{2\beta d}{d-1}\right\}$ 
and such that $\mu \geq \max\left\{20 \alpha, \frac{10\alpha d}{d-1}\right\}$.
Let $T(|Q|,k)$ be the maximum running time of the algorithm on an instance 
with size $|Q|$ and parameter $k$. To complete the running time analysis we will prove the following claim. 

\begin{claim}\label{clm:mainRuntime} $T(|Q|,k) \leq \mu \cdot f(k) \cdot k \cdot |Q|^\gamma$.
\end{claim}

\begin{proof}
We prove the claim by induction on $|Q|+k$. We will regularly make use of the facts that $f(k-1) \leq \frac{f(k)}{d}$ and that $f(k) \geq k$. We consider the execution of the algorithm on an instance ($Q=(D,R_1,\dots,R_\ell),k$). We need to prove that the running time of the algorithm is upper bounded by $\mu \cdot f(k) \cdot k \cdot |Q|^\gamma$. For the base cases if every strongly connected component in $D$ induces a substructure of $Q$ that is already in $\cQ$ or $k=0$, then the statement of the claim is satisfied by the choice of $\mu$. 
We now proceed to prove the inductive step. We will assume throughout the argument that $k \geq 1$ and that $Q\notin \cQ$.

If $D$ is not strongly connected then the algorithm makes two recursive calls; one to $Q_1=(Q[C], k - 1)$ and one to $Q_2=(Q - C, k - 1)$. Observe that $|Q_1|+|Q_2|\leq |Q|$. 
In this case the total time of the algorithm is upper bounded by
\begin{align*}
\beta k^2|Q| & +T(|Q_1|, k-1) + T(|Q_2|, k - 1) +  \alpha f(k) |Q|^\gamma \cdot 2k  \\
& \leq \frac{d-1}{2d}\cdot \mu \cdot f(k) \cdot k \cdot |Q|^\gamma 
+ \mu f(k-1) \cdot (k-1) \cdot |Q|^\gamma + \frac{d-1}{2d}\cdot \mu \cdot f(k) \cdot k \cdot |Q|^\gamma  \\
&  \leq \frac{d-1}{2d}\cdot \mu \cdot f(k) \cdot k \cdot |Q|^\gamma 
+ \mu \frac{f(k)}{d} \cdot k \cdot |Q|^\gamma + \frac{d-1}{2d}\cdot \mu \cdot f(k) \cdot k \cdot |Q|^\gamma \\
& \leq \mu \cdot f(k) \cdot k \cdot |Q|^\gamma \cdot \left(\frac{d-1}{2d} + \frac{d-1}{2d} + \frac{1}{d} \right)  \\
& = \mu \cdot f(k) \cdot k \cdot |Q|^\gamma 
\end{align*}
We will now assume in the rest of the argument that $D$ is strongly connected. For $k \geq 1$ and strongly connected $D$ the algorithm invokes Lemma~\ref{lem:crux}. Following the execution of the algorithm of Lemma \ref{lem:crux}, we execute the steps corresponding to exactly \emph{one} of the 4 cases. We show that in each of the four cases, the algorithm runs within the claimed time bound. Let $S$ be the set output by the algorithm of Lemma~\ref{lem:crux}. We now proceed with the case analysis.
	
\begin{description}
\item[Case 1:] In this case the algorithm terminates after one execution of Algorithm $\bA$ with a set $\hat{W}$ of size at most $3k$. Thus the total running time of the algorithm is upper bounded by 
\begin{align*}
\beta k^2|Q| + \alpha f(k) |Q|^\gamma \cdot 3k & \leq
\frac{1}{20}\mu \cdot f(k) \cdot k \cdot |Q|^\gamma + 
\frac{3}{20} \mu \cdot f(k) \cdot k \cdot |Q|^\gamma \\
& \leq \mu \cdot f(k) \cdot k \cdot |Q|^\gamma
\end{align*}

\item[Case 2:] In this case the algorithm makes two recursive calls, 
one to $(Q[C], k - 1)$ and one to $(Q - C, k - 1)$. After this, the algorithm executes Algorithm $\bA$ with a set $\hat{W}$ of size at most $5k$ and terminates. Let $Q_1 =Q[C]$ and $Q_2=Q-C$.   In this case the total time of the algorithm is upper bounded as follows.
\begin{align*}
\beta k^2|Q| & + T(|Q_1|, k-1) + T(|Q_2|, k-1) + \alpha f(k) |Q|^\gamma \cdot 5k \\ 
& \leq \frac{d-1}{2d}\cdot \mu \cdot f(k) \cdot k \cdot |Q|^\gamma 
+ \mu f(k-1) \cdot (k-1) \cdot |Q|^\gamma + \frac{d-1}{2d}\cdot \mu \cdot f(k) \cdot k \cdot |Q|^\gamma  \\
& = \mu \cdot f(k) \cdot k \cdot |Q|^\gamma 
%
\end{align*}

\item[Case 3:] In this case the algorithm makes a single recursive call on the instance $(Q[C],k)$, where $Q[C]$ has size at most $\frac{1}{2}|Q|$. 
 After the recursive call the algorithm executes Algorithm $\bA$ with a set $\hat{W}$ of size at most $4k$ and terminates. Hence, in this case the total time of the algorithm is upper bounded as follows.
\begin{align*}
\beta k^2(|Q|) &  + T\left(\frac{1}{2}|Q|, k\right) + \alpha f(k) |Q|^\gamma \cdot 4k \\
& \leq \frac{1}{20}\cdot \mu \cdot f(k) \cdot k \cdot |Q|^\gamma 
+ \frac{1}{2} \cdot \mu f(k) \cdot k \cdot |Q|^\gamma 
+ \frac{4}{20} \mu \cdot f(k) \cdot k \cdot |Q|^\gamma \\
& \leq \mu \cdot f(k) \cdot k \cdot |Q|^\gamma
\end{align*}

\item[Case 4:] Here the algorithm makes a single recursive call on $(Q-S, k-1)$. Following the recursive call, there is a single call to Algorithm $\bA$ with a set $\hat{W}$ of size at most $4k$. This yields the following bound on the running time in this case.
\begin{align*}
\beta k^2|Q| &  + T(|Q|, k-1) + \alpha f(k) |Q|^\gamma \cdot 4k \\
& \leq \frac{d-1}{2d}\cdot \mu \cdot f(k) \cdot k \cdot |Q|^\gamma 
+ \mu f(k-1) \cdot (k-1) \cdot |Q|^\gamma + \frac{d-1}{2d}\cdot \mu \cdot f(k) \cdot k \cdot |Q|^\gamma  \\
& \leq \mu \cdot f(k) \cdot k \cdot |Q|^\gamma
\end{align*}
\end{description}

In each of the four cases the running time of the algorithm, and hence $T(|Q|,k)$ is upper bounded by $\mu \cdot f(k) \cdot k \cdot |Q|^\gamma$. This completes the proof of the claim.
\end{proof}

The algorithm and its correctness proof, together with Claim~\ref{clm:mainRuntime} completes the proof of Theorem \ref{thm:meta}.

%
%
%
%
%
%

%% file: applications.tex
In this section, we describe how Theorem \ref{thm:meta} can be invoked to shave off a factor of $n$ from existing iterative compression based algorithms
for {\dfvs}, {\sc Directed  Feedback Arc Set} ({\dfas}), {\sc Directed Subset Feedback Vertex Set} and {\sc Multicut}. Here, {\dfas} is the \emph{arc} deletion version of {\dfvs} where the objective is to delete at most $k$ arcs from the given digraph to make it acyclic.
 
\myparagraph{1. Application to {\dfvs}.} 
 We set $\epsilon=1$ and define $\cal Q$ to be the set of all directed acyclic graphs. Clearly, $\cQ$ is linear-time recognizable, hereditary and rigid.  The algorithm $\bB$ is defined to be an algorithm that, given as input a digraph $D$ which is not acyclic, simply picks an arc $(a,b)$ which is part of a directed cycle in $D$ and returns $u,v$ where $u=b$ and $v=a$. The algorithm $\bA$ can be chosen to be any compression routine for {\dfvs}. In particular, we choose the compression routine of Chen et al.~\cite{ChenLLOR08} which runs in time $\bigoh(f(k)(n+m)\cdot |W|)$ where $f(k)=4^kk!k^4$. Invoking Theorem \ref{thm:meta} for {\sc {$\cQ$}-Deletion ($1$)}, we obtain our linear-time algorithm for {\dfvs}.

 \dfvstheorem*
 
 It is easy to see that {\dfas} can be reduced to {\dfvs} in the following way. For an instance $(D,k)$ of {\dfas},  subdivide each arc, and make $k+1$ copies of the original vertices to obtain a graph $D'$. It is straightforward to see that $(D,k)$ is a {\Yes} instance of {\dfvs} if and only if $(D',k)$ is a {\Yes} instance of {\dfas}. Since $|D'|\leq 2(k+1)|D|$, we also obtain a linear-time FPT algorithm for {\dfas}.
 
 \begin{corollary}
There is an algorithm for \dfas\ running in time $\Oh(k!4^kk^6\cdot (n+m))$. 	
 \end{corollary}

 \myparagraph{2. Application to {\sc Multicut}.}
 In the {\sc Multicut} problem, the input is an undirected graph $G$, integer $k$ and pairs of vertices $(s_1,t_1),\dots, (s_r,t_r)$ and the objective is to check whether there is a set $X$ of at most $k$ vertices such that for every $i\in [r]$, $s_i$ and $t_i$ are in different connected components of $G-X$. The parameterized complexity of this problem was open for a long time until Marx and Razgon \cite{MarxR14} and Bousquet, Daligault and Thomasse \cite{BousquetDT11} showed it to be {\FPT}. Marx and Razgon obtained their {\FPT} algorithm via the iterative compression technique. They 
 gave an algorithm for the compression version of {\sc Multicut} that, on input $D,(s_1,t_1),\dots, (s_r,t_r),k$ and $\hat W$,  runs in time 
  $2^{\bigoh(k^3)}\cdot n^\gamma\cdot |\hat W|$ for some $\gamma$. As a result, they were able to obtain an algorithm for {\sc Multicut} that runs in time $2^{\bigoh(k^3)}\cdot n^{\gamma+1}$. We show by an application of Theorem \ref{thm:meta} that we can improve  this running time by a factor of $n$.

 We define $\cQ$ to be the set of all pairs $(D,S)$ with $D$ being a digraph where $(u,v)$ is an arc if and only if $(v,u)$ is an arc ($D$ is essentially an undirected graph with all edges replaced with arcs in both directions),  $S\subseteq V(D)^2$ and for every $(u,v)\in S$, the vertices $u$ and $v$ are in different strongly-connected components of $D$. Clearly, $\cQ$ is  linear-time recognizable, hereditary and rigid.
 We define $\bA$ to be the compression routine of Marx and Razgon \cite{MarxR14} and $\bB$ to be an algorithm that computes the strongly connected components of $D$ and simply returns a pair $(u,v)\in S$ such that $u$ and $v$ are in the same strongly connected component of $D$. By invoking Theorem \ref{thm:meta} for {\sc $\cQ$-Deletion(2)} with these parameters, we obtain the following corollary.

 \begin{corollary}\label{cor:multicut}
 There is an algorithm for {\sc Multicut} running in time $2^{\bigoh(k^3)}n^\gamma$. 	
 \end{corollary}

 We remark that  since the objective of Marx and Razgon in their paper was to show the fixed-parameter tractability of {\sc Multicut}, they did not try to optimize $\gamma$.
   However, going through the algorithm of Marx and Razgon and making careful (but standard) modifications of the derandomization step in their algorithm using Theorem 5.16   \cite{CyganFKLMPPS14} (see also  \cite{AlonYZ95}) as well as the more recent linear time {\FPT} algorithms for the {\sc Almost 2-SAT} problem \cite{RamanujanS14,IwataOY14} instead of the algorithm in \cite{LokshtanovNRRS14}, it is possible to bound the running time of their compression routine and consequently the running time given in Corollary \ref{cor:multicut} by   $2^{O(k^3)}mn \log n$.

 \myparagraph{3. Application to {\sc Directed Subset Feedback Vertex Set}.}
 In the {\sc Directed Subset Feedback Vertex Set} (DSFVS) problem, the input is a digraph $D$, a set $S$ of vertices in $D$ and the objective is to check whether $D$ contains a vertex set $X$ of size at most $k$ such that $D-X$
  has no cycles passing through $S$, also called $S$-cycles. 
  This problem is a clear generalization of {\dfvs} and was shown to be {\FPT} by Chitnis et al. \cite{Chitnis:2012DSFVS} via the iterative compression technique. 
  
  They also observed that this problem is equivalent to the {\sc Arc Directed Subset Feedback Vertex Set} (ADSFVS) where the input is a digraph $D$ and a set $S$ of \emph{arcs} in $D$ and the objective is to check whether $D$ contains a vertex set $X$ of size at most $k$ such that $D-X$
  has no cycles passing through $S$.
   Chitnis et al. gave an algorithm for the compression version of ADSFVS that, on input $D,S,k$ and $\hat W$,  runs in time 
  $2^{\bigoh(k^3)}\cdot n^\gamma\cdot |\hat W|$ for some $\gamma$. As a result, they were able to obtain an algorithm for ADSFVS that runs in time $2^{\bigoh(k^3)}\cdot n^{\gamma+1}$. We show by an application of Theorem \ref{thm:meta} that we can directly shave off a factor of $n$ from this running time.

  We first argue that ADSFVS is a special case of {\sc {\cal Q}-Deletion ($2$)}.
  We define by $\cQ$ the set of all pairs $(D,S)$ where $S\subseteq A(D)$ and $D$ has no cycle passing through an arc in $S$. Clearly, $\cQ$ is linear-time recognizable, hereditary and rigid.
  We define $\bA$ to be the compression routine of Chitnis et al. \cite{Chitnis:2012DSFVS} and $\bB$ to be an algorithm that, given as input the pair $(D,S)$, computes the strongly connected components of $D$ and simply returns an arc in $S$ which is contained in a strongly connected component of $D$. 
  By invoking Theorem \ref{thm:meta} for {\sc {$\cQ$}-Deletion ($2$)} with these parameters, we  obtain the following corollary.
  
 \begin{corollary}
 	There is an algorithm for {\sc Arc Directed Subset Feedback Vertex Set} running in time $2^{\bigoh(k^3)}\cdot n^\gamma$.
 \end{corollary}
 
 Due to the aforementioned observation of Chitnis et al., we also get an algorithm with the same running time for {\sc Directed Subset Feedback Vertex Set}.

%% file: section-mainlemma.tex
In this section we will prove our main technical lemma, Lemma \ref{lem:crux}. For the sake of convenience, we restate it here.

\cruxlemma*


Before we proceed to the proof of the lemma, we need to set up some notation and recall known results on separators in digraphs. For the rest of this section, we fix $\epsilon\in {\mathbb N}$ and a linear-time recognizable hereditary and rigid family of $\epsilon$-structures, $\cQ$ and deal with this family. Furthermore, we will assume that all $\epsilon$-structures we deal with are of the same \emph{type} as $\cQ$.

\subsection{Setting up separator definitions}

\begin{definition}
	Let $D$ be a digraph and $X$ and $Y$ be disjoint vertex sets. A vertex set $S$ disjoint from $X\cup Y$ is called an $X$-$Y$ \textbf{separator} if there is no $X$-$Y$ path in $D-S$. We denote by $R(X,S)$ the set of vertices of $D-S$ reachable from vertices of $X$ via directed paths and by $NR(X,S)$ the set of vertices of $D-S$ \emph{not} reachable from vertices of $X$. We denote by $\lambda_D(X,Y)$ the size of a smallest $X$-$Y$ separator in $D$ with the subscript ignored if the digraph is clear from the context.
\end{definition}

We remark that it is not necessary that $Y$ and $N^+[X]$ be disjoint in the above definition. If these sets do intersect, then there is no $X$-$Y$ separator in the digraph and we define $\lambda(X,Y)$ to be $\infty$.

\begin{definition}
Let $D$ be a digraph and $X$ and $Y$ be disjoint vertex sets. Let $S_1$ and $S_2$ be $X$-$Y$ separators. We say that $S_2$ \textbf{covers} $S_1$  if $R(X,S_2)\supseteq R(X,S_1)$.
\end{definition}

Note that for a set $S\subseteq V(D)$ which is an $X$-$Y$ separator in $D$ for some $X,Y\subseteq V(D)$ the sets $R(X,S)$, $NR(X,S)$ and $S$ form a partition of the vertex set of $D$.

\begin{definition}
Let $Q$ be an $\epsilon$-structure and let $D$ be the digraph in $Q$ where $D$ is  strongly connected and let $u,v\in V(D)$. Let $S\subseteq V(D)$ be a $u$-$v$ separator in $D$. Then, we say that $S$ is  

\begin{itemize}
	\item an \textbf{l-good} $u$-$v$ separator if the induced substructure $Q[R(u,S)]\in \cQ$ and the induced substructure $Q[NR(u,S)]\notin \cQ$.
	\item an \textbf{r-good} $u$-$v$ separator if the induced substructure $Q[R(u,S)]\notin \cQ$ and the induced substructure $Q[NR(u,S)]\in \cQ$.
		\item a \textbf{dual-good} $u$-$v$ separator if the induced substructure $Q[R(u,S)]\in \cQ$ and the induced substructure $Q[NR(u,S)]\in \cQ$.
	\item a \textbf{completely-good} $u$-$v$ separator if the induced substructure $Q[R(u,S)]\in \cQ$ and the induced substructure $Q[NR(u,S)]\in \cQ$.
	\item an \textbf{l-light} $u$-$v$ separator if $|Q[R(u,S)]|\leq \frac{1}{2}|Q|$. 
	\item an \textbf{r-light} $u$-$v$ separator if $|Q[NR(u,S)]|\leq \frac{1}{2}|Q|$. 
\end{itemize}

\end{definition}

\begin{figure}[t]
\begin{center}
  \includegraphics[height=220pt, width=250pt]{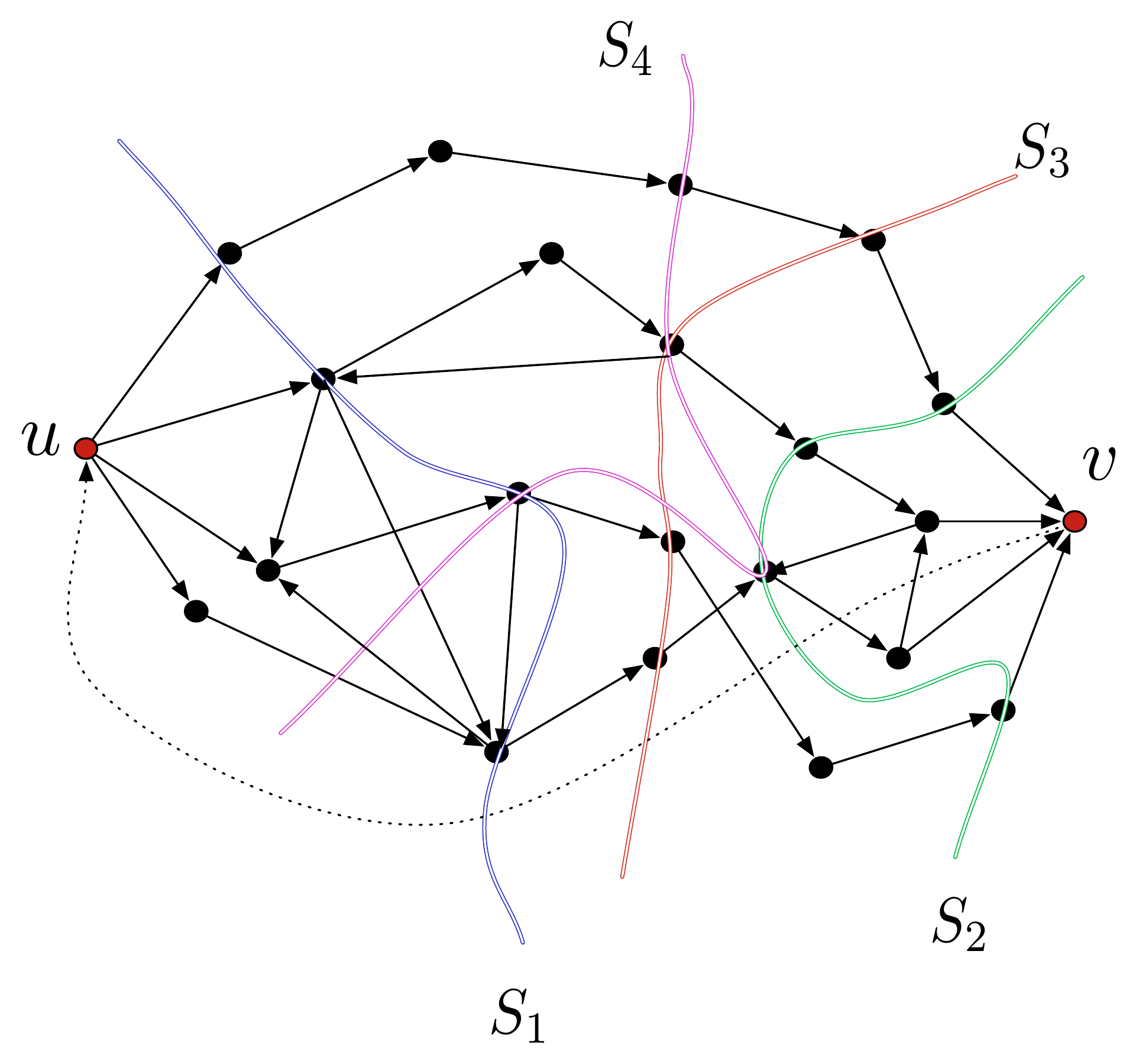}
  \caption{An illustration of the various $u$-$v$ separator types for the case of {\dfvs} (see Section \ref{sec:app} for a formal reduction of {\dfvs} to {\sc $\cQ$-Deletion ($r$)}). Here, $S_1$ is $l$-good, $S_2$ is $r$-good, $S_3$ is dual-good and $S_4$ is completely-good.}
  \label{fig:separatortypes}
  \end{center}
\end{figure}

Observe that the first 4 types in the above definition partition the set of $u$-$v$ separators. 
On the other hand, while any $u$-$v$ separator must be either $l$-light or $r$-light, it is possible that the same $u$-$v$ separator is both $l$-light and $r$-light. That is, the last 2 types cover but not necessarily partition the set of $u$-$v$ separators.
See Figure \ref{fig:separatortypes} for an illustration of separators of various types in the special case of $\cQ$ denoting the set of acyclic digraphs. We will prove Lemma \ref{lem:crux} by examining the interactions between separators of different types.

%

The next lemma  shows that a pair of separators in $D$ with one covering the other have a certain monotonic dependency between them regarding their ($l$/$r$)-goodness and ($l$/$r$)-lightness.

\begin{lemma}[{\bf Monotonicity Lemma}]\label{lem:monotone} Let $Q$ be an $\epsilon$-structure and let $D$ be the digraph in $Q$ where $D$ is strongly connected. Let  $u,v\in V(D)$ and let $S_1$ and $S_2$ be a pair of $u$-$v$ separators in $D$ such that $S_2$ covers $S_1$. Furthermore, suppose that neither $S_1$ nor $S_2$ is dual-good or completely-good. Then the following statements hold.
	
	\begin{itemize}
		\item If $S_1$ is $r$-good then $S_2$ is also $r$-good.
		\item If $S_2$ is $l$-good then $S_1$ is also $l$-good.
		\item If $S_1$ is r-light then $S_2$ is also r-light.
		\item If $S_2$ is l-light then $S_1$ is also l-light.
	\end{itemize}

\end{lemma}

\begin{proof}

We begin by proving the first statement of the lemma. Suppose to the contrary that $S_1$ is $r$-good and $S_2$ is $l$-good.
By definition, the substructure $Q_1=Q[R(u,S_1)]$ is not in $\cQ$ and $Q_2=Q[R(u,S_2)]$ is in $\cQ$.	
However, since $S_2$ covers $S_1$, we know that $R(u,S_2)\supseteq R(u,S_1)$. This implies that $Q_1$ is a substructure of $Q_2$. Since $\cQ$ is hereditary, we know that if $Q_1$ is not in $\cQ$, then neither is $Q_2$, a contradiction. This completes the proof of the first statement. The proofs of the remaining statements are all analogous.
\end{proof}

	We now prove the following lemma which  provides a linear time-testable sufficient condition for a separator to reduce the size of the solution upon deletion.

\begin{lemma}\label{lem:split_solution}
	Let $Q$ be an $\epsilon$-structure and let $D$ be the digraph in $Q$ where $D$ is strongly connected. Let $u,v\in V(D)$, $k\in {\mathbb N}$ and suppose that every deletion set of $Q$ into $\cQ$ hits all $u$-$v$ paths in $D$. Let $Z$ be an $r$-good (l-good) $u$-$v$ separator of size at most $k$ such that there is no $u$-$v$ separator of size at most $k$ contained entirely in the set $R(u,Z)$
	  (respectively $NR(u,Z)$). If $Q$ has a deletion set into $\cQ$ of size at most $k$ disjoint from $\{u,v\}$ then $Q-Z$ has a deletion set into $\cQ$ of size at most $k-1$.
	
\end{lemma}

\begin{proof}
Let $X$ be a deletion set of $Q$ into $\cQ$. 
Consider the case when $Z$ is an $r$-good separator. The argument for the other case is analogous.
Since $Z$ is $r$-good, we know that the substructure $Q[NR(u,Z)]$ is in $\cQ$. Therefore, any strongly connected component in the digraph $D-Z$ which induces a substructure \emph{not} in $\cQ$ lies in the set $R(u,Z)$. Also, the set $X'=X\cap R(u,Z)$ is by definition a deletion set for the substructure $Q[R(u,Z)]$. Since every strongly connected component of $D-Z$ which does not induce a substructure in $\cQ$ lies in the digraph $D[R(u,Z)]$, it follows that $X'$ is in fact a deletion set into $\cQ$ for $Q-Z$. We now claim that $X'\subset X$.

Suppose to the contrary that $X'=X$. 
By the premise of the lemma, we have that $X$ is a $u$-$v$ separator of size at most $k$.
Since $X'=X$, we conclude that $X$ is a $u$-$v$ separator of size at most $k$ which is contained in the set $R(u,Z)$, a contradiction to the premise of the lemma, implying that $X'\subset X$. This completes the proof of the lemma.
\end{proof}

Having set up the definitions and certain properties of the separators we are interested in, we now define the notion of separator sequences and describe our linear time subroutines that perform certain computations that will be critical for the linear time implementation of our main algorithm.

\subsection{Finding useful separators}

We begin with a lemma which gives a polynomial time procedure to compute, for every pair of vertices $s$ and $t$ in a digraph, a sequence of vertex sets each containing $s$ and excluding $t$ such that every minimum $s$-$t$ separator is contained in the union of the out-neighborhoods of these sets. Moreover, for each set, the out-neighborhood is in fact a minimum $s$-$t$ separator. The statement of this lemma is almost identical to the statements of Lemma 2.4 in \cite{MarxOR13} and Lemma 3.2 in \cite{RamanujanS14}. However, the statement of Lemma 2.4 in \cite{MarxOR13} deals with undirected graphs while that of Lemma 3.2 in \cite{RamanujanS14} deals with arc-separators instead of vertex separators. Furthermore, the second property in the statement of the following lemma is not part of the latter, although a closer inspection of the proof shows that this property is indeed guaranteed. Note that this proof closely follows that in \cite{MarxOR13}. We give a full proof here for the sake of completeness.

\begin{lemma}
\label{lem:separator layer}
 Let $s,t$ be two vertices in a digraph $D$ such that the minimum size of an $s$-$t$ separator is $\ell>0$. Then, there is an ordered collection ${\cal X}=\{X_1,\dots,X_q\}$ of vertex sets where $\{s\}\subseteq X_i\subseteq V(D)\setminus (\{t\}\cup N^-(t))$ such that
 
 \begin{enumerate}
\setlength{\itemsep}{-2pt}
  \item $X_1\subset X_2\subset \cdots \subset X_q$,
  \item $X_i$ is reachable from $s$ in $D[X_i]$ and every vertex in $N^+(X_i)$ can reach $t$ in $D-X_i$,
  \item $\vert N^+(X_i)\vert=\ell$ for every $1\leq i\leq q$ and 
  \item every $s$-$t$ separator of size $\ell$ is fully contained in $\bigcup_{i=1}^q N^+(X_i)$.
 \end{enumerate}
 
 \noindent
 Furthermore, there is an algorithm that, given $k\in {\mathbb N}$, runs in time $\bigoh(k(\vert V(D)\vert +\vert A(D)\vert))$ and either correctly concludes that $\ell>k$ or produces the sets $X_1,X_2\setminus X_1,\dots, X_q\setminus X_{q-1}$ corresponding to such a collection $\cal X$.

\end{lemma}

\begin{proof}
	

	We denote by $D'$ the directed network obtained from $D$ by performing the following operation.	Let $v\in V(D)\setminus \{s,t\}$. We remove $v$ and add 2 vertices $v^+$ and $v^-$. For every $u\in N^-(v)$, we add an arc $(u,v^-)$ of infinite capacity and for every $u\in N^+(v)$, we add an arc $(v^+,u)$ of infinite capacity and finally we add the arc $(v^-,v^+)$ with capacity 1. We now make an observation relating $s$-$t$ arc-separators in $D'$ to $s$-$t$ separators in $D$. But before we do so, we need to formally define arc-separators.

	\begin{definition}
	Let $D$ be a digraph and $s$ and $t$ be distinct vertices. An arc-set $S$ is called an $s$-$t$ {\bf arc-separator} if there is no $s$-$t$ path in $D-S$. We denote by $R(s,S)$ the set of vertices of $D-S$ reachable from $s$ via directed paths and by $NR(s,S)$ the set of vertices of $D-S$ \emph{not} reachable from $s$. 
\end{definition}

The following observation is a consequence of the definition of arc-separators and the construction of $D'$.
	
	\begin{observation}\label{obs:correspondence arc-vertex}If $S\subseteq \{(v^-,v^+)| v\in V(D)\setminus \{s,t\}\}$ is an $s$-$t$ arc-separator in $D'$, then the set $S^{-1}=\{v|(v^-,v^+)\in S\}$ is an $s$-$t$ separator in $D$. Conversely for every $s$-$t$ separator $X$ in $D$, the set $\{(v^-,v^+)|v\in X\}$ is an $s$-$t$ arc-separator in $D'$.
	\end{observation}

	We now proceed to the proof of the lemma statement.
 We first run $min\{k+1,\ell\}$ iterations of the Ford-Fulkerson algorithm \cite{FordFulkerson56} on the network $D'$. Since we do not know $\ell$ to begin with, we simply try to execute $k+1$ iterations. If we are able to execute $k+1$ iterations, then it must be the case that $\ell>k$ and hence we return that $\ell>k$. Otherwise, we stop after at most $\ell\leq k$ iterations with a maximum $s$-$t$ flow. Let $D_1$ be the residual graph. Let $C_1,\dots, C_q$ be a topological ordering of the strongly connected components of $D_1$ such that $i<j$ if there is a path from $C_i$ to $C_j$. Recall that there is a $t$-$s$ path in $D_1$. Let $C_x$ and $C_y$ be the strongly connected components of $D_1$ containing $t$ and $s$ respectively. Since there is a path from $t$ to $s$ in $D_1$, it must be the case that $x<y$. For each $x<i\leq y$, let $Y_i=\bigcup_{j=i}^qC_j$ (see Figure~\ref{fig:marxlemmapic}). We first show that $\vert \delta_{D'}^+(Y_i)\vert=\ell$ for every $x<i\leq y$. Since no arcs leave $Y_i$ in the graph $D_1$, no flow enters $Y_i$ and every arc in $\delta_{D'}^+(Y_i)$ is saturated by the maximum flow. Therefore, $\vert\delta_{D'}^+(Y_i)\vert=\ell$. 
 
 \begin{figure}[t] \begin{center}
\includegraphics[height=180 pt, width=340 pt]{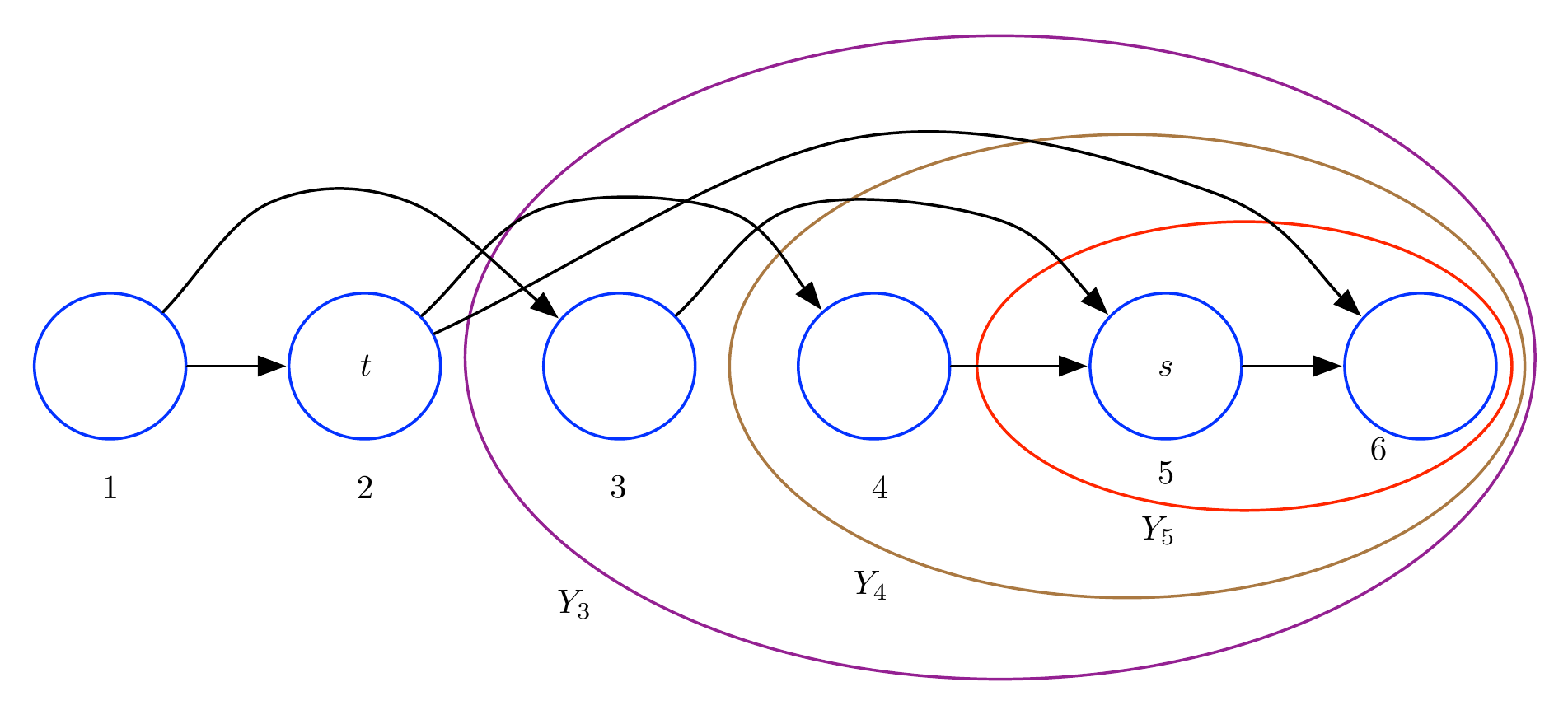}

 \caption{An illustration of the sets in the proof of Lemma~\ref{lem:separator layer}. The chain of circles in the middle are the strongly connected components of $D_1$ and $\alpha(s)=5$ and $\alpha(t)=2$.}
\label{fig:marxlemmapic}
\end{center}
\end{figure}
 
 We now show that every arc which is part of a minimum $s$-$t$ arc-separator is contained in $\bigcup_{i=1}^q \delta_{D'}^+(Y_i)$. Consider a minimum $s$-$t$ arc-separator $S$ and an arc $(a,b)\in S$. Let $Y$ be the set of vertices reachable from $s$ in $D'-S$. Since $F$ is a minimum $s$-$t$ arc-separator, it must be the case that $\delta_{D'}^+(Y)=F$ and therefore, $\delta_{D'}^+(Y)$ is saturated by the maximum flow. Therefore, we have that $(b,a)$ is an arc in $D_1$. Since no flow enters the set $Y$, there is no cycle in $D_1$ containing the arc $(b,a)$ and therefore, if the strongly connected component containing $b$ is $C_{i_b}$ and that containing $a$ is $C_{i_a}$, then $i_b<i_a$. Furthermore, since there is flow from $s$ to $a$ from $b$ to $t$, it must be the case that $x<i_b<i_a<y$ and hence the arc $(a,b)$ appears in the set $\delta_{D'}^+(Y_{i_a})$.
 
 Finally, we define the set $R(Y_i)$ to be the set of vertices of $Y_i$ which are reachable from $s$ in the graph $D'[Y_i]$. For each set $R(Y_i)$ we define the set $R^{-1}(Y_i)$ as $\{v|\{v^+,v^-\}\subseteq R(Y_i)\}$.
Due to the correspondence between $s$-$t$ separators in $D$ and $s$-$t$ arc-separators in $D'$ (Observation \ref{obs:correspondence arc-vertex}), the sets 
 $R^{-1}(Y_{y})\subset R^{-1}(Y_{y-1}) \subset \cdots \subset R^{-1}(Y_{x+1})$ indeed form a collection of the kind described in the statement of the lemma. It remains to describe the computation of these sets.

 In order to compute these sets, we first need to run the Ford-Fulkerson algorithm for $\ell$ iterations and perform a topological sort of the strongly connected components of $D_1$. This takes time $\bigoh(\ell(\vert V(D)\vert+\vert A(D)\vert))$. During this procedure, we also assign indices to the strongly connected components in the manner described above, that is, $i<j$ if $C_i$ occurs before $C_j$ in the topological ordering. In $\bigoh(\ell(\vert V(D)\vert+\vert A(D)\vert))$ time, we can assign indices to vertices such that the index of a vertex $v$ (denoted by $\alpha(v)$) is the index of the strongly connected component containing $v$. We then perform the following preprocessing for every vertex $v$ such that $\alpha(v)<\alpha(s)$. We go through the list of in-neighbors of $v$ and find 
$$\beta(v)=\max_{u\in N^-(v)} \Big\{ \alpha(u)~\vert~ \alpha(v)< \alpha(u) \Big\} \mbox{ and }$$  
$$\gamma(v)=\min_{u\in N^-(v)} \Big\{ \alpha(u)~\vert~ \alpha(v)< \alpha(u)\Big\}$$ 
and set $\beta'(v)=\min\{\beta(v),\alpha(s)\}$ and $\gamma'(v)=\max\{\gamma(v),\alpha(t)+1\}$.

 The meaning of these numbers is that the vertex $v$ occurs in each of the sets $N^+(Y_{\beta'(v)})$,  $N^+(Y_{\beta'(v)-1})$, $\dots$, $N^+(Y_{\gamma'(v)})$. 
  This preprocessing can be done in time $\bigoh(m+n)$ since we only compute the maximum and minimum in the adjacency list of each vertex. 
 A vertex $v$ is said to be $i$-forbidden for all $\gamma'(v)\leq i\leq \beta'(v)$. We now describe the algorithm to compute the sets in the collection.

\medskip

\noindent
{\bf Computing the collection.} We do a modified (directed) breadth first search (BFS) starting from $s$ by using only \emph{out-going} arcs. Along with the standard BFS queue, we also maintain an additional \emph{forbidden} queue.

We begin by setting $i=\alpha(s)$ and start the BFS by considering the out-neighbors of $s$.
We add a vertex to the BFS queue only if it is both unvisited and not $i$-forbidden. 
If a vertex is found to be $i$-forbidden (and it is not already in the forbidden queue), we add this vertex to the forbidden queue. Finally, when the BFS queue is empty and every unvisited out-neighbor of every vertex in this tree is in the forbidden queue, we return the set of vertices \emph{added} to the BFS tree in the current iteration as $R(Y_{i})\setminus R(Y_{i+1}) $.   Following this, the vertices in the forbidden queue which are not $(i-1)$-forbidden are removed and added to the BFS queue and the algorithm continues after decreasing $i$ by 1. 
The algorithm finally stops when $i=\alpha(t)$.

We claim that this algorithm returns each of the sets $R(Y_{\alpha(s)})$, $R(Y_{\alpha(s)-1})\setminus R(Y_{\alpha(s)})$, $\dots,R(Y_{\alpha(t)+1})\setminus R(Y_{\alpha(t)+2})$ and runs in time $\bigoh(\ell(\vert V(D)\vert +\vert A(D)\vert)$. In order to bound the running time, first observe that the vertices which are $i$-forbidden are exactly the vertices in the set 
$N^+(R(Y_i))$
 and therefore the number of $i$-forbidden vertices for each $i$ is at most $\ell$. This implies that the number of vertices in the forbidden queue at any time is at most $\ell$. Hence, testing if a vertex is $i$-forbidden or already in the forbidden queue for a fixed $i$ can be done in time $\bigoh (\ell)$. Therefore, the time taken by the algorithm is $\bigoh(\ell)$ times the time required for a BFS in $D$, which implies a bound of $\bigoh(\ell(\vert V(D)\vert +\vert A(D)\vert))$.

For the correctness, we prove the following invariant for each iteration. Whenever a set is returned in an iteration,
\begin{itemize} \setlength{\itemsep}{-2pt}
\item the set of vertices currently in the forbidden queue are exactly the $i$-forbidden vertices
\item the vertices in the current BFS tree are exactly the vertices in the set $R(Y_{i})$. 
\end{itemize}

\noindent
For the first iteration, this is clearly true. We assume that the invariant holds at the end of iteration $j\geq 1$ (where $i=i'$) and consider the $(j+1)$-th iteration (where $i$ is now set as $i'-1$).

Let $P_j$ be the vertices present in the BFS tree at the end of the $j$-th iteration and $P_{j+1}$ be the vertices present in the BFS tree at the end of the $(j+1)$-th iteration. We claim that the set $P_{j+1}=R(Y_{i'-1})$.
 
Since we never add a vertex to $P_{j+1}$ if it is $(i'-1)$-forbidden, the vertices in $P_{j+1}\setminus P_j$ are precisely those vertices which are reachable from $P_j$ via a path disjoint from $(i'-1)$-forbidden vertices.
Since the invariant holds for the preceding iteration, we know that $P_j=R(Y_{i'})$ and by our observation about $P_{j+1}\setminus P_j$, we have that $P_{j+1}$ is the set of vertices reachable from $R(Y_{i'})$ via paths disjoint from $(i'-1)$-forbidden vertices, which implies that $P_{j+1}=R(Y_{i'-1})$ since $R(Y_{i'-1})$ is precisely the set of vertices reachable from $R(Y_{i'})$ via paths disjoint from $(i'-1)$-forbidden vertices. We now show that the vertices in the forbidden queue are exactly the $(i'-1)$-forbidden vertices.
Since the BFS tree in iteration $j+1$ could not be grown any further, every out-neighbor of every vertex in the tree is in the forbidden queue. Since we have already shown that the vertices in the BFS tree, that is in $P_{j+1}$, are precisely the vertices in $R(Y_{i'-1})$, we have that every $(i'-1)$-forbidden vertex is already in the forbidden queue.
This proves that the invariant holds in this iteration as well and completes the proof of correctness of the algorithm.
\end{proof}

We also require the following well known property of minimum separators. This is a simple consequence of  Property 4 in Lemma \ref{lem:separator layer}.

%
%
%
%
%
%
%

\begin{lemma}\label{lem:easy application}
Let $D$ be a digraph and $s,t$ be two vertices. Let ${\cal X}=\{X_1,\dots, X_q\}$ be the collection given by Lemma \ref{lem:separator layer} and $\ell=|N^+(X_i)|$ for each $i\in [q]$. Define $X_0=\emptyset$ and $X_{q+1}=V(D)$. Let $Z_i$ denote the set $X_{i+1}\setminus N^+[X_{i}]$ for each $0\leq i\leq q$. Then, any minimal $s$-$t$ separator in $D$ that intersects $Z_i$ for any $0\leq i\leq q$ has size at least $\ell+1$. 
\end{lemma}

\begin{proof} 
Let $Q=\bigcup_{j=1}^q N^+(X_j)$. We claim that for any $0\leq i\leq q$, the set $Z_i$ is disjoint from $Q$. Fix an index $i$ and consider a vertex $u\in Z_i$. By definition, $u\in X_{i+1}$ and $u\notin N^+[X_{i}]$. Since $u\in X_{i+1}$, it must be the case that $u\in X_{r}$ and hence \emph{not in} $N^+[X_r]$ for every $r>i$ (by Property 1 in Lemma \ref{lem:separator layer}). Similarly, since $u\notin N^+[X_i]$, it must be the case that $u\notin N^+[X_r]$ for any $r\leq i$. Therefore, $u\notin Q$ and we conclude that $Z_i$ is disjoint from $Q$.

The lemma now follows from the fact that $Z_i$ is disjoint from $Q$ and Property 4 in Lemma \ref{lem:separator layer} which guarantees that every $s$-$t$ separator of size $\ell$ is contained in $Q$. This completes the proof of the lemma.	
\end{proof}

We now recall the notion of a tight separator sequence. This was first defined in \cite{LokshtanovR12} for undirected graphs. Here we define a similar notion for directed graphs.

%
%
%
%
%

\begin{definition}\label{def:smallest separator sequence} 

Let $s$,$t$ be two vertices in a digraph $D$ and let $k\in {\mathbb N}$. A \textbf{tight $s$-$t$ separator sequence} of order $k$ is an ordered collection ${\cal H}=\{H_1,\dots,H_q\}$ of sets in $V(D)$ where $\{s\}\subseteq H_i\subseteq V(D)\setminus (\{t\}\cup N^-(t))$ for any $1\leq i\leq q$ such that,

\begin{itemize}
	\item $H_1\subset H_2\subset \dots \subset H_q$,
	\item $H_i$ is reachable from $s$ in $D[H_i]$ and every vertex in $N^+(H_i)$ can reach $t$ in $D-H_i$ \\
	(implying that $N^+(H_i)$ is a minimal $s$-$t$ separator in $D$)
	\item $\vert N^+(H_i)\vert\leq k$ for every $1\leq i\leq q$, 
	\item for any $1\leq i\leq q-1$, there is no $s$-$t$ separator $S$ of size at most $k$ where $S\subseteq H_{i+1}\setminus N^+[H_{i}]$ or $S\cap N^+[H_q]=\emptyset$.
\end{itemize}

\end{definition}

We have the following obvious but useful consequence of the definition of tight separator sequences.

\begin{lemma}
	\label{lem:induced tight sequence}Let $s$,$t$ be two vertices in a digraph $D$ and let $k\in {\mathbb N}$. Let $u\in V(D)$ be a vertex which is part of \emph{every} minimal $s$-$t$ separator of size at most $k$. Then, ${\cal H}$ is a tight $s$-$t$ separator sequence of order $k$ in $D$ if and only if it is a tight $s$-$t$ separator sequence of order $k-1$ in $D-\{u\}$. Furthermore, $u\in N^+(H)$ for every $H\in {\cal H}$.
	\end{lemma}

The following lemma gives a linear-time {\FPT} algorithm to compute a tight separator sequence for a given parameter $k$. In fact, it is a \emph{polynomial} time algorithm which depends linearly on the input size while the dependence on the parameter is a polynomial. This subroutine plays a major role in the proof of Lemma \ref{lem:crux}. 

\begin{lemma}\label{lem:find_tight_sequence}
	There is an algorithm that, given a digraph $D$ with no isolated vertices, vertices $s,t\in V(D)$ and $k\in {\mathbb N}$, runs in time $\bigoh(k^2m)$ and either correctly concludes that there is no $s$-$t$ separator of size at most $k$ in $D$ or returns the sets $H_1,H_2\setminus H_1,\dots, H_q\setminus H_{q-1}$ corresponding to a tight $s$-$t$ separator sequence ${\cal H}=\{H_1,\dots,H_q\}$ of order $k$.
\end{lemma}

\begin{proof} 
The algorithm we present executes the algorithm of Lemma \ref{lem:separator layer} on various carefully chosen subdigraphs of the given graph and 
 Lemma \ref{lem:easy application} allows us to prove a bound on the number of times any single arc of $D$ participates in these computations.


Suppose that $\lambda(s,t)=\ell<k$ and consider the output of the algorithm of Lemma \ref{lem:separator layer} on input $D$, $s$ and $t$.
%
%
By definition, this invocation returns the sets $X_1,X_2\setminus X_1, X_q\setminus X_{q-1}$ corresponding to the collection ${\cal X}=\{X_1,\dots, X_q\}$. We define $X_{q+1}$ to be the set $R(s,\emptyset)\setminus \{t\}$.
%
%
%
%
%
We set $X_0=\emptyset$ and for each $1\leq i\leq q+1$, we define the following sets (see Figure \ref{fig:localsets}) :

\begin{figure}[t]
\begin{center}
  \includegraphics[height=200pt, width=300pt]{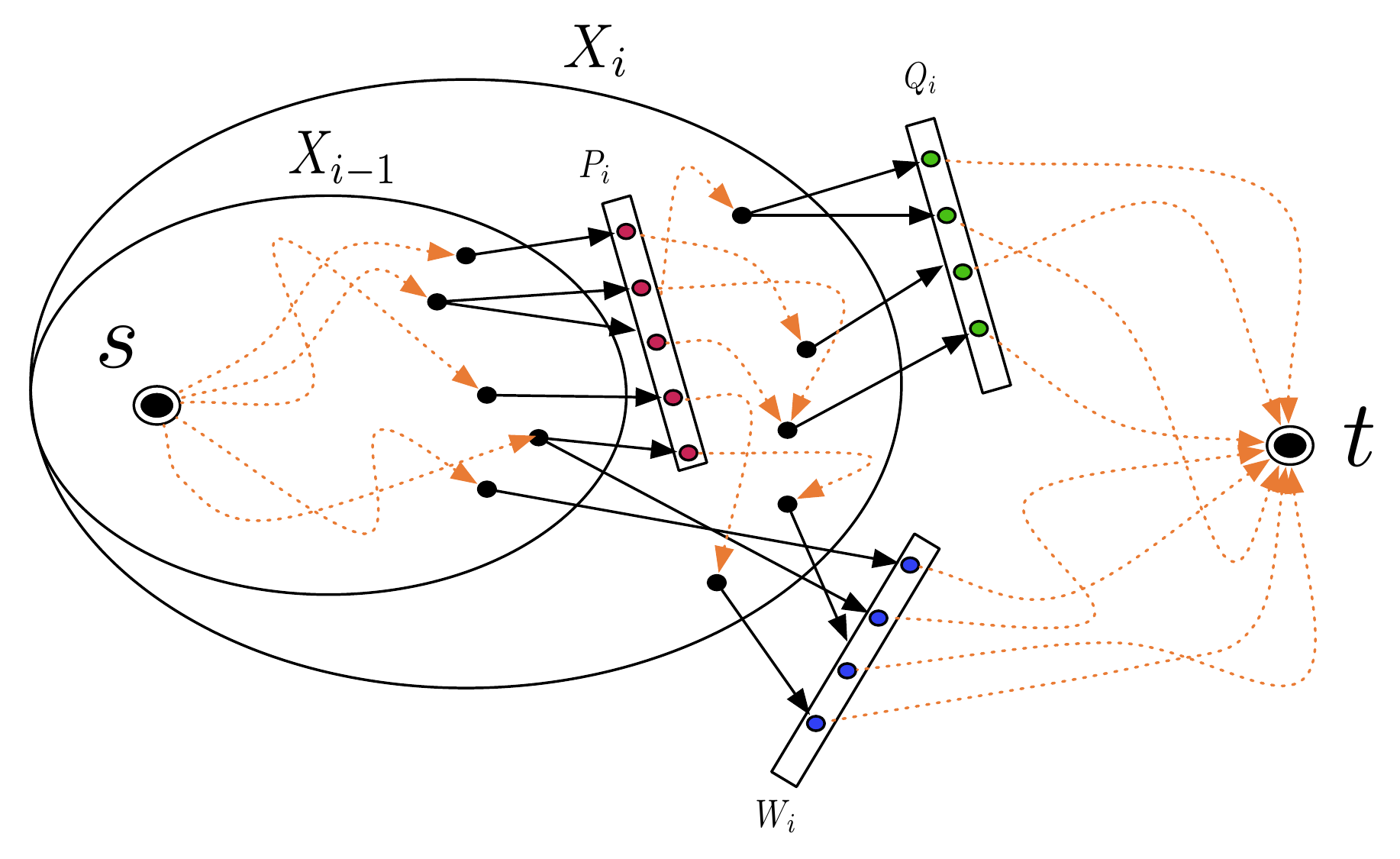}
  \caption{An illustration of the various sets defined in the proof of Lemma \ref{lem:find_tight_sequence}. The dotted arrows denote directed paths while the solid ones denote arcs.}
  \label{fig:localsets}
  \end{center}
\end{figure}


\begin{itemize}

\item $Y_i=X_i\setminus X_{i-1}$
\item $ P_i=Y_i\cap N^+(X_{i-1})$
\item $Q_i=N^+(X_i)\setminus N^+(X_{i-1})$
\item $W_i=N^+(X_i)\setminus Q_i$
\end{itemize}

with $P_1=\{s\}$.
That is, $P_i$ is defined to be those vertices in $Y_i$ (which is non-empty due to Property 1 in Lemma \ref{lem:separator layer}) which are out-neighbors of vertices in $X_{i-1}$, $Q_i$ is the set of those vertices in the out-neighborhood of $X_i$ which are \emph{not} in the out-neighborhood of $X_{i-1}$ and $W_i$ is the set of vertices in the out-neighborhood of $X_i$ which are not already in $Q_i$. Observe that $Q_i$ can also be written as $Q_i=(V(D)\setminus X_i)\cap (N^+(Y_i)\setminus N^+(X_{i-1}))$. Also note that $P_i$ and $Q_i$ are by definition disjoint.
Furthermore, it is important to note that $P_i$ and $Q_i$ are non-empty. The set $P_i$ is non-empty because Property 1 of Lemma \ref{lem:separator layer} guarantees that the set $Y_i$ is non-empty and Property 2 of Lemma \ref{lem:separator layer} ensures that every vertex in $X_i$ (and hence in $Y_i$) is reachable from $s$ in $D[X_i]$ implying that there is at least one vertex in $Y_i$ which has a vertex in $X_{i-1}$ as an in-neighbor. On the other hand, if $Q_i$ is empty then $N^+(X_i)=W_i$ and $N^+(X_{i-1})\supset W_i$ (strict superset since $P_i$ is non-empty). This contradicts Property 3 of Lemma \ref{lem:separator layer}. Finally, note that $P_1=\{s\}$, $Q_{q+1}=\{t\}$, $W_1=W_{q+1}=\emptyset$ and $P_{q+1}=N^+(X_q)$.
For each $1\leq i\leq q+1$ we also define the digraph $D_i$ as follows:

$$ V(D_i)= (Y_i\setminus P_i)\cup \{s_i,t_i\} \cup W_i$$
$$ \hspace{-50 pt} A(D_i)=A(D)[Y_i\setminus P_i] $$ 
$$ \hspace{120 pt} \bigcup \{(s_i,p)|p\in (N^+(P_i)\cap (Y_i\setminus P_i))\cup W_i\}$$ $$\hspace{60 pt} \bigcup \{(p,t_i)|p\in N^-(Q_i)\cup W_i\}$$

Finally, if $Q_i\cap N^+(P_i)\neq \emptyset$, then we add an arc $(s_i,t_i)$.
That is, the digraph $D_i$ is defined as the digraph obtained from $D[Y_i\cup Q_i]$ by adding the vertices in $W_i$, identifying the vertices of $P_i$ into a single vertex called $s_i$ (removing self-loops and parallel arcs), identifying the vertices of $Q_i$ into a single vertex called $t_i$ and adding arcs from $s_i$ to all vertices in $W_i$ and from all vertices in $W_i$ to $t_i$. Since $P_i$ and $Q_i$ are disjoint and non-empty, this digraph is well-defined. Also note that there is no isolated vertex in $D_i$. This is because every vertex in $D_i$ is reachable from $s_i$ by definition.
We now make the following claim regarding the connectivity from $s_i$ to $t_i$ in the digraph $D_i$.

\begin{claim}\label{clm:first}
For each $1\leq i\leq q+1$, $\lambda_{D_i}(s_i,t_i)>\ell$.	
\end{claim}

\begin{proof} Observe that if $Q_i\cap N^+(P_i)\neq \emptyset$, then by definition the graph $D_i$ contains the arc $(s_i,t_i)$, implying that $\lambda_{D_i}(s_i,t_i)=\infty$. Henceforth, we assume that $Q_i\cap N^+(P_i)= \emptyset$. 
Consider a set $S_i\in V(D_i)$ that is an $s_i$-$t_i$ separator in $D_i$. Observe that since $S_i$ is disjoint from $\{s_i,t_i\}$, it must be the case that $S_i\subseteq V(D)$. Furthermore, observe that by definition, $S_i\supseteq W_i$. This is because each vertex in $W_i$ is both an out-neighbor of $s_i$ and an in-neighbor of $t_i$. We claim that $S_i$ intersects all $P_i$-$Q_i$ paths in $D$.

    Suppose that this is not the case and there is a $P_i$-$Q_i$ path in $D-S_i$. Let $J$ be a  $P_i$-$Q_i$ path in $D-S_i$ which minimizes the intersection with $P_i\cup Q_i$. 
    As a result of the minimality condition, it must be the case that this path begins at a vertex $p\in P_i$, ends at a vertex $p'\in Q_i$ and has all internal vertices in the set $Y_i$. 
    However, a corresponding $s_i$-$t_i$ path $J'$ in $D_i$ can be obtained  by simply replacing $p$ with $s_i$ and $p'$ with $t_i$. Since $J'$ is disjoint from $S_i$, we get a contradiction to our assumption that $S_i$ is an $s_i$-$t_i$ separator in $D_i$. Hence, we conclude that $S_i$ intersects all $P_i$-$Q_i$ paths in $D$.
    
   Now, observe that any $s$-$t$ path in $D$ that is disjoint from $W_i$ must contain as a subpath a $P_i$-$Q_i$ path whose internal vertices lie entirely in $Y_i$. This is because $s\in X_{i-1}$ and $N^+[X_i]$ is disjoint from $N^-[t]$ (guaranteed by Lemma \ref{lem:separator layer}). Since $S_i$  intersects all such paths, we conclude that $S_i$ is in fact an $s$-$t$ separator in $D$. 
   
   Furthermore, the presence of a $P_i$-$Q_i$ path in $D$ with all internal vertices in $Y_i$ and the fact that $S_i$ is a set disjoint from $P_i\cup Q_i$ that intersects this path implies that $S_i$ contains a vertex in $Y_i\setminus P_i=X_i\setminus N^+[X_{i-1}]$. 
   But notice that $S_i$ is an $s$-$t$ separator in $D$ that satisfies the premise of Lemma \ref{lem:easy application}. Hence, we conclude that $|S_i|>\ell$. This completes the proof of the claim. \end{proof}

The above claim allows us to recursively apply our algorithm to compute tight separator sequences on each graph $D_i$ while Claim \ref{clm:first} guarantees a bound on the depth of this recursion. The next claim shows that once we recursively compute a tight separator sequence in each of these digraphs, there is a linear time procedure to combine these sequences to obtain a tight separator sequence in the original graph.

\begin{claim}\label{clm:second}
For each $1\leq i\leq q+1$, let ${\cal L}^i$ denote a tight $s_i$-$t_i$ separator sequence $\{L^i_1,L^i_2,\dots, L^i_{r_i}\}$ of order $k$ in the digraph $D_i$. For each $1\leq i\leq q+1$ and $1\leq j\leq r_i$, let $H^i_j$ denote the set $(L^i_j\setminus \{s_i\})\cup P_i$.
Then, the ordered collection ${\cal H}$ defined as $X_0\cup H^1_1,\dots,$ $X_0\cup H^1_{r_1}, X_1, X_1\cup H^2_1,\dots, X_1\cup H^2_{r_2},\dots$,  $X_{q}, X_{q}\cup H^{q+1}_1,\dots, X_{q}\cup H^{q+1}_{r_{q+1}}$ is a tight $s$-$t$ separator sequence of order $k$ in $D$.
	
\end{claim}

\begin{proof}

Observe that by definition, for each $1\leq i\leq q+1$ and $1\leq j\leq r_i$, the set $H^i_j$ is a subset of $V(D)$. We now proceed to argue that $\cal H$ is a tight $s$-$t$ separator sequence of order $k$ in $D$. In order to do so, we need to prove that it satisfies the 4 conditions in Definition \ref{def:smallest separator sequence}.

\begin{figure}[t]
\begin{center}
  \includegraphics[height=240pt, width=350pt]{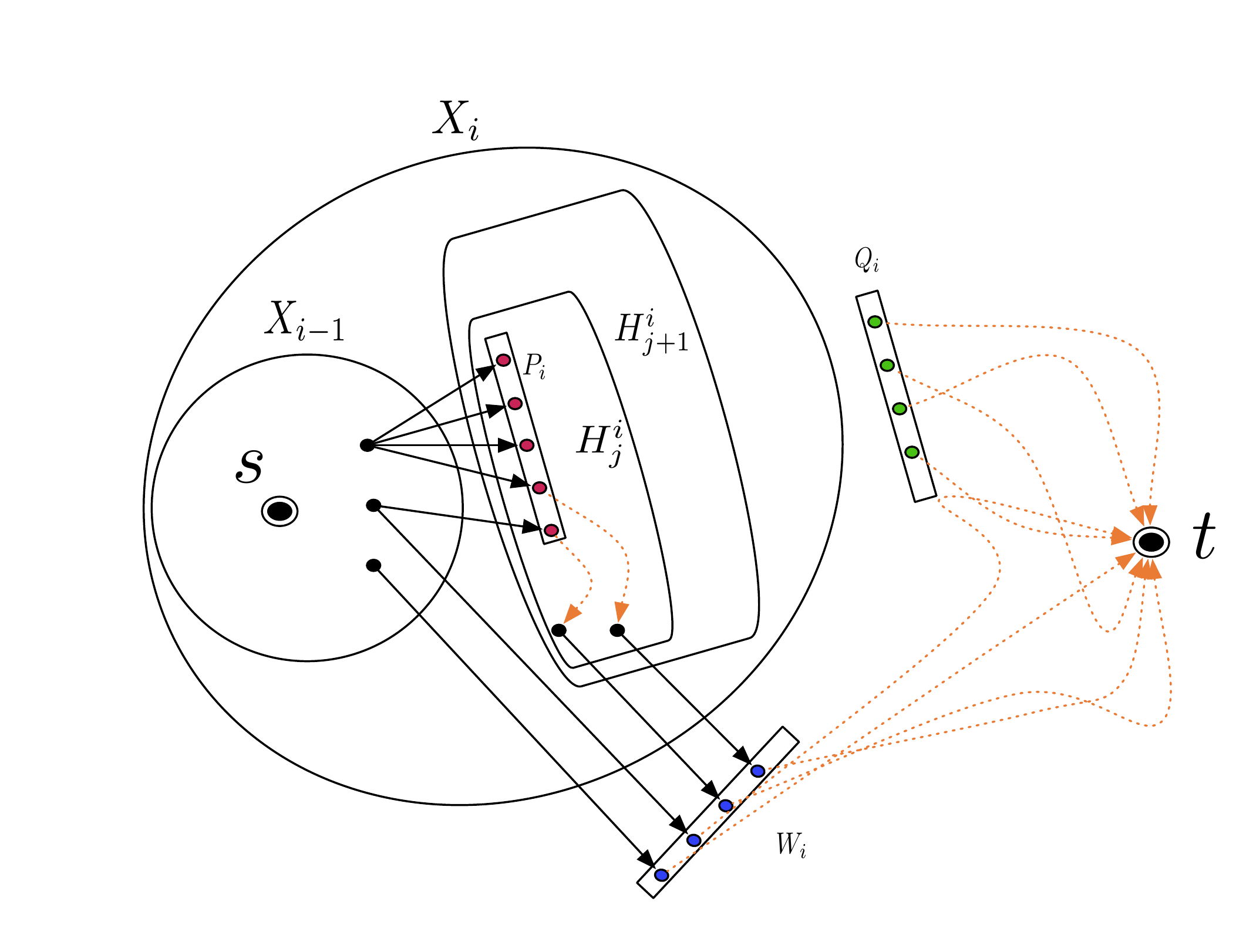}
  \caption{An illustration of the case where the both sets in $\cal H$ under consideration are contained in $Y_i$.}
  \label{fig:case1}
  \end{center}
\end{figure}

We begin by arguing that the collection satisfies the first condition. That is, for any two consecutive sets in $\cal H$ (recall that $\cal H$ is ordered and hence among any pair of consecutive sets there is a well-defined notion of first and second), the first set is a strict subset of the second. For this, we need to consider the following three cases. In the first case, there is an $1\leq i\leq q+1$ and a $1\leq j\leq r_i-1$ such that the two sets under consideration are $X_{i-1}\cup H^i_{j}$ and $X_{i-1}\cup H^i_{j+1}$ (see Figure \ref{fig:case1}). In this case, the property holds because $H^i_{j}\subset H^i_{j+1}$ by our assumption that ${\cal L}^i$ is a tight $s_i$-$t_i$ separator sequence in $D_i$. In the second case, there is an $1\leq i\leq q$ such that the two sets under consideration are $X_{i-1}\cup H^i_{r_i}$ and $X_i$.
In this case, since $X_i= X_{i-1}\cup Y_i$ and $H^i_{r_i}$ is a strict superset of $Y_i$ (by definition of the graph $D_i$), we know that $X_i\supset X_{i-1}\cup H^i_{r_i}$. In the third case, there is an $1\leq i\leq q$ such that the two sets under consideration are $X_i$ and $X_i\cup H^{i+1}_1$. Clearly, the second set contains the first. It only remains to argue that this containment is strict. Since we have already argued that $P_{i+1}$ is non-empty, we conclude that $H^{i+1}_1$ is also non-empty since it contains $P_{i+1}$. This in turn allows us to conclude that $X_i\cup H^{i+1}_{1}$ is a strict superset of $X_i$, thus completing the argument that $\cal H$ satisfies the first condition of Definition \ref{def:smallest separator sequence}.

We now move to the second condition. That is, for each set $H$ in $\cal H$, every vertex in $H$ is reachable from $s$ in the induced digraph $D[H]$ and every vertex in $N^+(H)$ can reach $t$ in $D-H$. Consider a set $H \in {\cal H}$. If $H=X_i$ for some $i\in [q]$, then we are already done since the required property is guaranteed by Lemma \ref{lem:separator layer} (second property). Therefore, we consider the case when there is $1\leq i\leq q+1,1\leq j\leq r_i$ such that $H=X_{i-1}\cup H^i_j$. We know by the second property of Lemma \ref{lem:separator layer} that every vertex in $X_{i-1}$ is reachable from $s$ in $D[X_{i-1}]$ and hence in $D[H]$. Furthermore, by definition, $H^i_j\supseteq P_i$ and the vertices reachable from $s_i$ in $D_i[L^i_j]$ are precisely those vertices in the set $\{x|$ $\exists y\in P_i: y {\xrightarrow\star} x$ in $D[H^i_j]\}$. Since every vertex in $P_i$ has an in-neighbor in $X_{i-1}$ by definition, we conclude that every vertex in $P_i$ is reachable from $s$ in the digraph $D[X_i\cup P_i]$ and since by the above observation every vertex in $H^i_j$ is reachable from a vertex in $P_i$ in the graph $D[H^i_j]$, we conclude that every vertex in $H$  is reachable from $s$ in the digraph $D[H]$.

 Similarly, by the second property of Lemma \ref{lem:separator layer}, every vertex in $N^+(X_i)$ can reach $t$ in $D-X$. Since $Q_i,W_i\subseteq N^+(X_i)$, we infer that every vertex in $Q_i,W_i$ can reach $t$ in $D-X_i$. 
 Furthermore, every vertex of $N^+(H^i_j)$ can reach $t_i$ in the digraph $D_i-H^i_j$. This is due to our assumption that ${\cal L}^i$ is a tight $s_i$-$t_i$ separator sequence in $D_i$. 
 Since we have already argued that every vertex in $W_i$ can reach $t$ in $D-X_i$, it is also the case that every vertex in $W_i$ can reach $t$ in $D-H$ as $H\subseteq X_i$. Since $W_i\subseteq N^+(X_i)$ by definition, we conclude that every vertex in the set $N^+(H)\cap W_i$ can reach $t$ in the digraph $D-H$. It remains to argue the same for the vertices in $N^+(H)\setminus W_i$.

Note that the set $N^+(H)\setminus W_i$ is contained in $Y_i$. If this were not the case and $N^+(H)\setminus W_i$ contained a vertex in $V(D)\setminus N^+[X_i]$ then any such vertex would be in $W_i$ by definition. On the other hand, if $N^+(H)\setminus W_i$ contained a vertex in $Q_i$ then the set $L^i_j$ would have $t_i$ as an out-neighbor, which is a contradiction to our assumption that ${\cal L}^i$ is a tight $s_i$-$t_i$ separator sequence in $D_i$.

But this implies that $N^+(H)\setminus W_i=N^+(H^i_j)=N^+(L^i_j)$. Observe that $N^+_{D_i}(L^i_j)=N^+_D(L^i_j)$ and hence we ignore the explicit reference to the digraph in which we consider the out-neighborhood of $L^i_j$. Since ${\cal L}^i$ is a tight $s_i$-$t_i$ separator sequence in $D_i$, we know that every vertex in $N^+(L^i_j)$ can reach $t_i$ in $D_i-L^i_j$. Since $H^i_j=(L^i_j\setminus \{s_i\}) \cup P_i$, we have that $N^+(H^i_j)=N^+(L^i_j)$ and therefore, 
every vertex in $N^+(H^i_j)$ can reach $Q_i$ in the digraph $D[Y_i\cup Q_i]-H^i_j$. Combining this with the fact that every vertex in $Q_i$ can reach $t$ in $D-X_i$, we conclude that every vertex in $N^+(H)\setminus W_i$ also can reach $t$ in $D-H$. This completes the argument for the second condition.

%
%

For the third condition, we need to argue that for each $H\in {\cal H}$ the size of the set $N^+(H)$ is at most $k$. Again, for each $H\in {\cal H}$ such that $H=X_i$ for some $i\in [q]$, we are already done. Now, consider a set $H \in {\cal H}$ and let $1\leq i\leq q-1,1\leq j\leq r_i$ be such that $H=X_{i-1}\cup H^i_j$. Recall that $W_i$ is the set of those vertices in $N^+(X_{i-1})$ that are not in $Y_i$ and hence $W_i\subseteq N^+[H]$. Also, for any vertex  $u\in N^+(H)$ such that $u\notin W_i$, it must be the case that $u$ is already in  $N^+(H^i_j)$. However, it follows from the definition of $D_i$ and $H^i_j$ that the set $N^+(L^i_j)$ is in fact the same as $W_i\cup H^i_j$. This implies that $N^+(H)=N^+(L^i_j)$ which has size at most $k$ due to our assumption that ${\cal L}^i$ is a tight $s_i$-$t_i$ separator sequence of order $k$ in $D_i$.

The final condition has two parts. First, we need to prove that for any 2 consecutive sets $H_1$ and $H_2$ in $\cal H$ (where $H_1$ appears before $H_2$ in the ordered collection), there is no $s$-$t$ separator of size at most $k$ that is contained in the set $H_2\setminus N^+[H_1]$. Secondly, we also need to prove that there is no $s$-$t$ separator of size at most $k$ that is disjoint from $N^+[X_q\cup H^{q+1}_{r_{q+1}}]$. We begin with the first part. Since $H_1$ occurs before $H_2$ in the ordering, our earlier arguments guarantee that $H_1\subset H_2$. 
Let $S$ be an arbitrary set contained in $H_2\setminus N^+[H_1]$ of size at most $k$. We will argue that $S$ cannot be an $s$-$t$ separator in $D$. Let $i\in [q+1]$ be the least value such that $S\subseteq X_i\setminus X_{i-1}$. The definition of ${\cal H}$ guarantees the existence of such an $i$. The definition of the sets $P_i,Q_i$ and the digraph $D_i$ implies that if $S$ is an $s$-$t$ separator in $D$ then it is an $s_i$-$t_i$ separator in $D_i$. We now consider the following three cases for the sets $H_1$ and $H_2$ and assuming that $S$ is an $s_i$-$t_i$ separator in $D_i$, obtain a contradiction in each case.

In the first case, 
there is a $1\leq j\leq r_i-1$ such that $H_1=X_{i-1}\cup H^i_{j}$ and $H_2=X_{i-1}\cup H^i_{j+1}$. In this case, we claim that $S$ is contained in the set $L^i_{j+1}\setminus N^+[L^i_j]$. Indeed, since $S\subseteq (H_2\setminus N^+[H_1])$ and $H_2\setminus N^+[H_1]$ is the same as $H^i_{j+1}\setminus N^+[H^i_j]$, we know that $S\subseteq L^i_{j+1}\setminus N^+[L^i_j]$. Thus we have concluded that $S$ is an $s_i$-$t_i$ separator in $D_i$ which is contained in the set $L^i_{j+1}\setminus N^+[L^i_j]$. This contradicts our assumption that ${\cal L}^i$ is a tight $s_i$-$t_i$ separator sequence of order $k$.

In the second case,  $H_1=X_{i-1}$ and $H_2=X_{i-1}\cup H^i_{1}$. In this case, the same argument as that above shows that $S$ is contained in the set $L^i_{1}$. However, our assumption that ${\cal L}^i$ is a tight $s_i$-$t_i$ separator sequence of order $k$ in $D_i$ implies that $S$ cannot be an $s_i$-$t_i$ separator, a contradiction.

In the third and final case, $H_1=X_{i-1}\cup H^i_{r_i}$ and $H_2=X_i$. In this case, observe that $S$ is disjoint from the set $N^+_{D_i}[L^i_{r_i}]$. But the second part of the final condition in Definition \ref{def:smallest separator sequence} applied to ${\cal L}^i$ implies that $S$ cannot be an $s_i$-$t_i$ separator in $D_i$, a contradiction. This completes the argument for the third case.

Having thus completed the argument for the first part of the final condition, we now conclude the proof of the claim by arguing the second part. That is, there is no $s$-$t$ separator of size at most $k$ disjoint from the set $N^+_D[X_q\cup H^{q+1}_{r_{q+1}}]$. Suppose to the contrary that $S$ is an $s$-$t$ separator of this kind. Clearly, $S$ is an $s_{q+1}$-$t_{q+1}$ separator of size at most $k$ in $D_{q+1}$. Moreover, since it is disjoint from $N^+_D[X_q\cup H^{q+1}_{r_{q+1}}]$, it follows that it is disjoint from $N^+_{D_{q+1}}[L^{q+1}_{r_{q+1}}]$. However, this contradicts our assumption that ${\cal L}^{q+1}$ is a tight $s_{q+1}$-$t_{q+1}$ separator sequence of order $k$ (by violating the second part of the fourth condition in Definition \ref{def:smallest separator sequence}).
This completes the proof of the claim.
\end{proof}

We now use the claims above to complete the proof of the lemma. We describe the complete algorithm.\\

\myparagraph{Description of the algorithm.}We begin by running the algorithm of Lemma \ref{lem:separator layer} on the graph $D$ with $s$ and $t$ the same as those in the premise of the lemma. If this subroutine concludes that there is no $s$-$t$ separator of size at most $k$ in $D$ then we return the same. Otherwise,  the subroutine returns the sets $X_1,X_2\setminus X_1, \dots, X_q\setminus X_{q-1}$ corresponding to the collection ${\cal X}=\{X_1,\dots, X_q\}$. We define $X_{q+1}$ to be the set $R(s,\emptyset)\setminus \{t\}$.

 Having computed the sets $X_1,\dots, X_{q+1}$, for each $1\leq i\leq q+1$ we compute the graph $D_i$, and recursively compute the sets $L^i_1,L^i_2\setminus L^i_1,\dots, L^i_{r_i}\setminus L^i_{r_{i}-1} $
corresponding to a  tight $s_i$-$t_i$ separator sequence ${\cal L}^i=\{L^i_1,L^i_2,\dots, L^i_{r_i}\}$ of order $k$ in the graph $D_i$. At this point, we note a subtle computational simplification we use. In order to compute ${\cal L}^i$, for those $D^i$s where $W_i\neq \emptyset$, we can invoke Lemma \ref{lem:induced tight sequence} and compute a tight $s_i$-$t_i$ separator sequence of order $k-|W_i|$ in the graph $D_i-W_i$. As a result, we never actually need to construct the entire graph $D_i$ as defined earlier. Instead it suffices to construct $D_i-W_i$. The reason behind this is that we can now consider the arcs in the graphs $D_1,\dots, D_{q+1}$ to be a partition of a  subset of the arcs in $D$.

For each $1\leq i\leq q+1$ and $1\leq j\leq r_i$, let $H^i_j$ denote the set $(L^i_j\setminus \{s_i\})\cup P_i$. We output the sets $H^1_1,H^1_2\setminus H^1_1, \dots, H^1_{r_1}\setminus H^1_{r_1-1}$, $X_1\setminus H^1_{r_1-1}$, $H^2_1, H^2_2\setminus H^2_1, \dots, H^2_{r_2}\setminus H^2_{r_2-1}$, $X_2\setminus H^2_{r_2-1}$, $\dots$
which correspond (by Claim \ref{clm:second}) to a tight $s$-$t$ separator sequence ${\cal H}=H^1_1,\dots,$ $H^1_{r_1}, X_1, X_1\cup H^2_1,\dots, X_1\cup H^2_{r_2},\dots$,  $X_{q}, X_{q}\cup H^{q+1}_1,\dots, X_{q}\cup H^{q+1}_{r_{q+1}}$ or order $k$. 
Since the correctness is a direct consequence of Claim \ref{clm:second}, we now proceed to the running time analysis.\\

\myparagraph{Running time.} We analyse the running time of this algorithm in terms of $k,m$ and $\lambda_D(s,t)$. We let $T(k,\lambda,m)$ denote the running time of the algorithm when $\lambda=\lambda_D(s,t)$. If $\lambda>k$, then $T(k,\lambda,m)=\bigoh(k m)$. This is because in this case, we only require a single execution of the algorithm of Lemma \ref{lem:separator layer} to conclude that $k<\lambda$. Otherwise, the description of the algorithm clearly implies the following recurrence.

$$T(k,\lambda,m)=\bigoh(\lambda m)+ \sum_{i=1}^{q+1}T(k,\lambda_i,m_i)$$

where $\lambda_i=\lambda_{D_i}(s_i,t_i)$ and $m_i$ denotes the number of arcs in $D_i$. Note that $m\geq\sum_{i=1}^{q+1}m_i$. 
The $\bigoh(\lambda m)$ term includes the time required to execute the algorithm of Lemma \ref{lem:separator layer} as well as the time required to compute the graphs $D_1,\dots, D_{q+1}$.
Now, due to Claim \ref{clm:first}, we have that $\lambda_i>\lambda$ for each $i\in [q+1]$. Unrolling the recurrence with $\lambda>k$ being the base case, the claimed running time follows. This completes the proof of the lemma.	
\end{proof}

Having completed the description of our main algorithmic subroutine, we now proceed to complete this subsection by describing a pair of subroutines performing computations on tight separator sequences. In this first lemma, we argue that given the output of Lemma \ref{lem:find_tight_sequence}, one can, in linear time find a pair of consecutive separators in the sequence where the first is l-light and the second one is not. The output of this lemma will form an `extremal' point of interest in our algorithm for {\sc $\cQ$-Deletion($\epsilon$)}.

\begin{lemma}\label{lem:linear_light}
Let $Q=(D,R_1,\dots, R_\ell)$ be an $\epsilon$-structure where $D$ is  strongly connected. Let $u,v\in V(D)$, $k\in {\mathbb N}$. 
 Let ${\cal H}=\{H_1,\dots,H_q\}$ be a tight $u$-$v$ separator sequence in $D$ with the algorithm of Lemma \ref{lem:find_tight_sequence} returning the sets $H_1,H_2\setminus H_1,\dots, H_q\setminus H_{q-1}$. There is an algorithm that, given $D,u,v,k$ and these sets, runs in time $\bigoh(k|Q|)$ and computes the least $i$ for which the separator $N^+(H_i)$ is l-light and the separator $N^+(H_i)$ is not l-light (and consequently is r-light)  or correctly concludes that there is no such $1\leq i\leq q$.	
\end{lemma}

\begin{proof}
	Given the sets $H_1,H_2\setminus H_1,\dots, H_q\setminus H_{q-1}$ we label the vertices of $V(D)$ in the following way with elements from $\{1,\dots,q\}$. We set $H_0=\{u\}$,  $H_{q+1}=V(D)$ and for each $i\in \{0,\dots,q\}$, we label the vertices of $H_{i+1}\setminus H_i$ with the label $i+1$. We denote the label of a vertex $w$ by $lab(w)$. Observe that any vertex with label $i$ has at most $k$ out-neighbors whose labels are greater than $i$. Therefore, for every vertex of label $i$, all but $k$ of its out-neighbors are labelled $i$ or less. Finally, we assume that for each $H_i$ we have marked the set of at most $k$ vertices in $N^+[H_i]$.	This can be done in time $\bigoh(m)$ by performing a directed bfs from $u$ and marking a vertex $w$ as being in the set $N^+[H_i]$ for the least $i$ such that $i$ is less than label of $w$ and $w$ has an in-neighbor with label $i$. It then follows from the definition of $\cal H$ that $w$ is in the set $N^+[H_j]$ for every $i\leq j\leq lab(w)$. This is the reason why we only keep track of the earliest $i$ for which $w\in N^+[H_i]$.
	We now proceed to design the claimed algorithm.
	
	We begin by iterating $i$ from 0 to $q$ and compute the number of arcs contained strictly inside each $H_i$, a number we denote by $L_i$. We do this as follows. Since $H_0=\{u\}$, $L_0$ is trivially 0. Therefore, we begin by examining the set $H_1$ and compute $L_1$. For any $i>1$, assuming we have already computed $L_{i-1}$, we now describe the computation of $L_i$. We iterate over the vertices in $H_{i+1}\setminus H_i$ and for each vertex $w$ in this set we count the number of arcs which have $w$ as a tail and have as the head any vertex except the at most $k$ of  $N^+(H_i)$ which have already been marked. If $w$ was a vertex marked as $N^+(H_{i-1})$ then we also count the set of arcs with $w$ as a head and having as a tail a vertex which is labelled at most $i-1$. It is clear that this algorithm computes the numbers $L_1,\dots, L_q$ correctly. 
	  Observe that every arc of $D$ is examined at most $2k$ times in the entire procedure. Thus, in time $\bigoh(km)$, we will have computed the size of the set $L_i$ for every $i$ from 0 to $q$.
	
	 For each $x\in [\ell]$ and $y\in \{0,\dots, q\}$, let $J^x_y$ denote the number of tuples of $R_x$ which are contained in the substructure of $Q$ induced by $H_y$. Clearly, if we compute all numbers $J^x_y$ in the required time, then the claimed algorithm follows. But recall that we have already labeled vertices of $H_{i+1}\setminus H_i$ with the label $i+1$ for every $i\in \{0,\dots,q\}$. Therefore for any $x\in [\ell]$ and any tuple in $R_x$, the least $y$ for which this tuple is to be counted towards $J^x_y$ is the \emph{largest} label among the elements in this tuple. Since this only requires a single linear search among the vertices in each tuple which requires a total time of $\bigoh(|Q|)$, the lemma follows. 	
	\end{proof}

The next lemma gives a linear time subroutine that checks whether the substructure induced on the set $H_2\setminus N^+[H_1]$ is in $\cQ$,  for a pair $H_1,H_2$ of consecutive sets in the tight separator sequence computed by the algorithm of Lemma \ref{lem:find_tight_sequence}.

\begin{lemma}
	\label{lem:check_cycle}
	Let $Q$ be an $\epsilon$-structure and let $D$ be the digraph in $Q$ where $D$ is  strongly connected. Let $u,v\in V(D)$ and $k\in {\mathbb N}$. Let ${\cal H}=\{H_1,\dots,H_q\}$ be a tight $u$-$v$ separator sequence in $D$ with the algorithm of Lemma \ref{lem:find_tight_sequence} returning the sets $H_1,H_2\setminus H_1,\dots, H_q\setminus H_{q-1}$. There is an algorithm that, given $Q,u,v,k$ and these sets, runs in time $\bigoh(k|Q|)$ and computes the least $i$ for which the substructure $Q[H_{i+1}\setminus N^+[H_{i}]]$ is not in $\cQ$ or correctly concludes that there is no such $1\leq i\leq q-1$.
\end{lemma}

\begin{proof} The proof of this lemma is similar to that of the previous lemma. Given the sets $H_1,H_2\setminus H_1,\dots, H_q\setminus H_{q-1}$ we label the vertices of $V(D)$ in the following way with elements from $\{1,\dots,q\}$. We set $H_0=\{u\}$,  $H_{q+1}=V(D)$ and for each $i\in \{0,\dots,q\}$, we label the vertices of $H_{i+1}\setminus H_i$ with the label $i+1$. For each $0\leq i\leq q$, we do a directed bfs/dfs on the set of vertices which are labeled $i$ but not marked as being part of the set $N^+(H_j)$ for some $j<i$. Since each arc is examined $\bigoh(k)$ times, the time bound follows.	
\end{proof}

Having set up all the required definitions as well as the subroutines needed for the main lemma, we now proceed to its proof.

\subsection{Proving the main lemma}

We complete this section by proving Lemma \ref{lem:crux}. We begin by restating the lemma here.

%

\cruxlemma*

\input{cruxproof.tex}

%% file: cruxproof.tex
\begin{proof}
We execute the algorithm of Lemma \ref{lem:find_tight_sequence} to either conclude that there is no $u$-$v$ separator of size at most $p$ or compute a tight $u$-$v$ separator sequence of order $p$.
	If this algorithm concludes that there is no $u$-$v$ separator of size at most $p$ in $D$, then we return the same.
	 Hence, we may assume that the subroutine returns sets $H_1,H_2\setminus H_1,\dots, H_q\setminus H_{q-1}$ corresponding to a tight $u$-$v$ separator sequence  ${\cal H}=\{H_1,\dots,H_q\}$ of order $p$.

We let $Z_i$ denote the set $N^+(H_i)$ for each $1\leq i\leq q$ and focus our attention on the sets $Z_1$ and $Z_q$ (which are not necessarily distinct). 
We begin by studying the set $Z_1$. If $Z_1$ 	
	 is dual-good  then setting $S=Z_1$ satisfies  Property 2. This is because we started with a strongly connected graph and by the definiton of dual-goodness both substructures $Q[R(u,Z_1)]$ and $Q[NR(u,Z_1)]$ are not in $\cQ$.
	 Similarly, if $Z_1$ is completely-good, then setting $S=Z_1$ satisfies Property 1.
	Now, suppose that $Z_1$ is $r$-good. It follows from Definition \ref{def:smallest separator sequence} that there is no $u$-$v$ separator of size at most $p$ contained entirely in the set $R(u,Z_1)$.
Then, by Lemma \ref{lem:split_solution}, if $Q$ has a deletion set into $\cQ$ of size at most $p$ disjoint from $\{u,v\}$ then $Q-Z_1$ has a deletion set into $\cQ$ of size at most $p-1$ and hence we set $S=Z_1\cup \{u,v\}$ and we satisfy Property 4. 
	Therefore, going forward, we assume that $Z_1$ is $l$-good. That is, the substructure $Q[H_1]$ is in $\cQ$. Note that given $Z_1$, this check can be performed in time $\bigoh(|Q|)$.	
	
	We have a symmetric argument for $Z_q$. That is, if $Z_q$ is dual-good or completely good then setting $S=Z_q$ satisfies Properties 2 and 1 respectively. Otherwise, if $Z_q$ is $l$-good, then by Definition \ref{def:smallest separator sequence} we know that there is no $u$-$v$ separator of size at most $p$ contained entirely in the set $NR(u,Z_q)$ and by Lemma \ref{lem:split_solution}, if $Q$ has a deletion set into $\cQ$ of size at most $p$ disjoint from $\{u,v\}$ then $Q-Z_q$ has a deletion set into $\cQ$ of size at most $p-1$ and hence we set $S=Z_q\cup \{v,u\}$ and we are done. Therefore, from this point on, we assume that $Z_q$ is $r$-good. That is, the substructure induced on $NR(u,Z_q)$ is in $\cQ$. Again checking which one of these cases hold can be done in time $\bigoh(|Q|)$.

	We now examine each of the sets $H_1,H_2\setminus H_1,\dots, H_q\setminus H_{q-1}$ and check if for any $i$, the digraph $D[H_{i+1}\setminus N^+[H_i]]$ is not in $\cQ$. This procedure can be performed in time $\bigoh(k|Q|)$ due to Lemma \ref{lem:check_cycle}. We now have 2 cases.

	In the first case, suppose that the subroutine returned an index  $1\leq i\leq q-1$ such that the substructure $Q[H_{i+1}\setminus N^+[H_i]]$ is not in $\cQ$. We now study the sets $Z_i$ and $Z_{i+1}$. By definition, it cannot be the case that $Z_i$ is $r$-good or $Z_{i+1}$ is $l$-good. Also, if either $Z_i$ or $Z_{i+1}$ is dual-good or completely-good (which can be checked in linear time) then we are done in a manner similar to that discussed earlier by setting $S=Z_i$ or $S=Z_{i+1}$. Hence, we may assume that $Z_i$ is $l$-good and $Z_{i+1}$ is $r$-good. Now, let $S=Z_i\cup Z_{i+1}\cup \{u,v\}$. Clearly, $|S|\leq 2p+2$. It remains to prove that $S$ satisfies one of the properties in the statement of the lemma. Precisely, we will prove that if $Q$ has a deletion set into $\cQ$ of size at most $p$ then  $Q-S$ has a deletion set into $\cQ$ of size at most $p-1$, that is, $S$ satisfies Property 4. 
	
	Let $X$ be a deletion set into $\cQ$ for $Q$  of size at most $p$. If $u\in X$ or $v\in X$, then we are already done. Therefore, assume that $u,v\notin X$. We claim that $X'=X\cap (H_{i+1}\setminus N^+[H_i])$ is in fact a deletion set into $\cQ$ for $Q-S$. This is because, any strongly connected component  in $D-S$ which induces a structure \emph{not in} $\cQ$ must be contained entirely within $R(u,Z_i)$ or $H_{i+1}\setminus N^+[H_i]$ or $NR(u,Z_{i+1})$. Since $Z_i$ is $l$-good and $Z_{i+1}$ is $r$-good, the substructures induced by the first and third sets are in $\cQ$. Therefore, any strongly connected component in $D-S$ which induces a structure not in $\cQ$ lies entirely in the set $H_{i+1}\setminus N^+[H_i]$. Since $X$ is a deletion set into $\cQ$ for $Q$, it follows that $X'$ is a deletion set into $\cQ$ for $Q-S$. We now claim that $X'\subset X$ and hence has size at most $|X|-1$. Suppose that this is not the case and  $X'=X$. By the premise of the lemma, we know that $X$ is a $u$-$v$ separator and hence we obtain a  contradiction to our assumption that $\cal H$ is a tight-separator sequence (violates condition 4 in Definition \ref{def:smallest separator sequence}). This is because $X$ itself will be a $u$-$v$ separator of size at most $p$ which is contained in the set $H_{i+1}\setminus N^+[H_i]$. This completes the argument for the case when the subroutine returns an  $1\leq i\leq q-1$ for which the substructure $Q[H_{i+1}\setminus N^+[H_i]]$ is not in $\cQ$. Henceforth, we will assume that for every $1\leq i\leq q-1$, the substructure $Q[H_{i+1}\setminus N^+[H_i]]$, denoted by $\hat Q_i$ is in $\cQ$.

	We now revisit the separators $Z_1$ and $Z_q$. Recall that $Z_1$ is $l$-good and $Z_q$ is $r$-good. Now, suppose that $Z_1$ is $r$-light. That is, the size of the substructure of $Q$ induced by the set $V(D)\setminus H_1$ is at most $\frac{1}{2}|Q|$. Then, we set $S=Z_1$. Observe that since $Q[H_1]$ is in $\cQ$, every  strongly connected component of $D-S$ which induces a structure not in $\cQ$ \emph{must} lie in the set $V(D)\setminus H_1$ and hence setting $Q_1=Q[V(D)\setminus H_1]$ and $Q_2=Q[H_1]$ satisfies Property 3. A symmetric argument holds if $Z_q$ is $l$-light. Therefore, we conclude that $Z_1$ is 
	     not $r$-light and $Z_2$ is not $l$-light. Therefore, $Z_1$ is 
	     $l$-light and $Z_2$ is $r$-light.

%
%

	Due to the monotonicity lemma (Lemma \ref{lem:monotone}), we know that there is an $i\geq 1$ such that $Z_i$ is $l$-light, $Z_{i+1}$ is not $l$-light (and so is $r$-light), and for all $j\leq i$, $Z_j$ is $l$-light and for all $j>i$, $Z_j$ is not $l$-light.
	We examine the sets in $\cal H$ and find this index $i$. That is, $Z_i$ is $l$-light and $Z_{i+1}$ is $r$-light. This can be done in linear time due to Lemma \ref{lem:linear_light}.


  If either of $Z_i$ or $Z_{i+1}$ is dual-good or completely good then we are done as argued earlier. So, we assume that each of $Z_i$ and $Z_{i+1}$ is either $l$-good or $r$-good.
     
     If $Z_{i+1}$ is $l$-good then setting $S=Z_{i+1}\cup \{v\}$ satisfies Property 3.
     Similarly, if $Z_{i}$ is $r$-good then setting $S=Z_{i}\cup \{v\}$ satisfies Property 3. It remains to handle the case when $Z_i$ is $l$-good and $Z_{i+1}$ is $r$-good. However, in this case, we claim that $Z_i\cup Z_{i+1}$ is in fact a deletion set for $Q$. Observe that any  strongly connected component of $D-(Z_i\cup Z_{i+1})$ which induces a structure not in $\cQ$ lies entirely in one of the sets $H_i$ or $H_{i+1}\setminus N^+[H_i]$ or $V(D)\setminus N^+[H_{i+1}]$. Since $\cQ$ is rigid, we only need to consider the strongly connected components of $D-(Z_i\cup Z_{i+1})$.
     The first and third sets induce  structures in $\cQ$ because $Z_i$ is $l$-good and $Z_{i+1}$ is $r$-good. The second set induces a structure in $\cQ$ because we have already argued that for every $1\leq j\leq q-1$, the substructure $Q[H_{j+1}\setminus N^+[H_j]]$, must be in $\cQ$.

 Therefore, we conclude that $Z_i\cup Z_{i+1}$ is a deletion set into $\cQ$ for $Q$ and setting $S=Z_i\cup Z_{i+1}\cup \{u,v\}$ satisfies Property 1. 
     This completes the proof of the lemma.
\end{proof}

%% file: conclusion.tex

We have presented the first linear-time {\FPT} algorithm for the classical {\dfvsfull} problem. For this, we introduced a new separator based iterative shrinking approach that either reduces the parameter or reduces the size of the instance by a constant fraction. We showed that our approach can be extended  to the directed version  of the {\sc Subset Feedback Vertex Set} ({\sc Subset FVS}) problem as well as to the {\sc Multicut} problem. As a result, any linear-time {\FPT} algorithm for the \emph{compression} version of these problems can be converted to one for the general problem as well. Furthermore, we have shown that \emph{any} further improvements in the running time of the compression routine for these problems can be directly lifted to the general problem.
Note that in the case of {\sfvs} on undirected graphs, the best known \emph{deterministic} algorithm already runs in time $2^{\bigoh(k \log k)}(m+n)$ \cite{LokshtanovRS15}. 

Finally, since our algorithm for {\dfvs} works via a black box reduction to the compression version of {\dfvs},  an algorithm with running time $2^{o(k \log k)}(n+m)$ for the compression version would immediately imply an algorithm with essentially the same time bound for {\dfvs}. We conclude with the following three open problems regarding \dfvs:  

\begin{itemize}
\setlength{\itemsep}{-2pt}
\item Does \dfvs\ admit a polynomial kernel?
\item Does \dfvs\ have an algorithm with running time $c^kn^{\Oh(1)}$?
\item Is {\sc Weighted Directed Feedback Vertex Set} fixed-parameter tractable?
\end{itemize}

\myparagraph{Acknowledgements.} The authors thank D\'{a}niel Marx for helpful discussions about the polynomial factor of the algorithm for {\sc Multicut} in~\cite{MarxR14}.